\documentclass[acmsmall,10pt]{acmart}

\acmJournal{PACMPL}
\acmVolume{1}
\acmNumber{OOPSLA} %
\acmArticle{1}
\acmYear{2018}
\acmMonth{1}
\acmDOI{} %
\startPage{1}

\setcopyright{none}
\renewcommand\footnotetextcopyrightpermission[1]{}

\usepackage{comment}
\usepackage{cite}
\usepackage{xspace}
\usepackage{commands}
\usepackage{language}
\usepackage{amsmath,amssymb,latexsym}
\usepackage{amsfonts}
\usepackage{mathtools}
\usepackage[inference]{semantic}
\usepackage{enumitem}
\usepackage{multirow}
\usepackage{tabularx}
\usepackage{balance}
\usepackage{hyperref}
\usepackage{enumitem}
\newlist{steps}{enumerate}{1}
\setlist[steps, 1]{label = Step \arabic*:}

\usepackage{listings}
\usepackage{amssymb}

\ifdefined\withcolor
	 \definecolor{haskellblue}{rgb}{0.0, 0.0, 1.0}
	 \definecolor{haskellstr}{rgb}{0.2, 0.2, 0.6}
	 \definecolor{haskellred}{rgb}{1.0, 0.0, 0.0}
  \definecolor{gray_ulisses}{gray}{0.55}
  \definecolor{castanho_ulisses}{rgb}{0.71,0.33,0.14}
  \definecolor{preto_ulisses}{rgb}{0.41,0.20,0.04}
  \definecolor{green_ulises}{rgb}{0.2,0.75,0}
\else
	\definecolor{haskellblue}{gray}{0.1}
	\definecolor{haskellstr}{gray}{0.1}
	\definecolor{haskellred}{gray}{0.1}
	\definecolor{gray_ulisses}{gray}{0.1}
	\definecolor{castanho_ulisses}{gray}{0.1}
	\definecolor{preto_ulisses}{gray}{0.1}
	\definecolor{green_ulisses}{gray}{0.1}
\fi

\definecolor{lcolor}{gray}{0.0}
\definecolor{lappcolor}{gray}{0.0}
\definecolor{lappascolor}{gray}{0.0}

\def\codesize{\normalsize}

\newcommand\showfocus[1]{\color{purple}{\textbf{#1}}}

\lstdefinelanguage{HaskellUlisses} {
	basicstyle=\ttfamily\small,
	moredelim=[is][\showfocus]{\#}{\#},
	sensitive=true,
	morecomment=[l][\color{gray_ulisses}\ttfamily\itshape\codesize]{--},
	morecomment=[s][\color{gray_ulisses}\ttfamily\itshape\codesize]{\{-}{-\}},
	morestring=[b]",
	stringstyle=\color{haskellstr},
	showstringspaces=false,
	numberstyle=\codesize,
	numberblanklines=true,
	showspaces=false,
	breaklines=true,
	showtabs=false,
  literate={ {quals}{{$\mathbb{Q}$}}1
             {iquals}{{$\mathbb{Q}$}}2
             {ltsolzero}{{$A_0$}}2
             {band}{{$\textbf{\texttt{and}}$}}2
             {->}{{$\rightarrow$}}2
             {not}{{$\neg\!\!\!$}}2
             {GType}{{\tilde{\mathit{\texttt{Type}}}}}4
           },
	emph=
	{[1]
		FilePath,IOError,abs,acos,acosh,all,and,any,appendFile,approxRational,asTypeOf,asin,
		asinh,atan,atan2,atanh,basicIORun,break,catch,ceiling,chr,compare,concat,concatMap,
		const,cos,cosh,curry,cycle,decodeFloat,denominator,digitToInt,div,divMod,drop,
		dropWhile,either,elem,encodeFloat,enumFrom,enumFromThen,enumFromThenTo,enumFromTo,
		error,even,exp,exponent,fail,mapMaybe,filter,flip,floatDigits,floatRadix,floatRange,floor,
		fmap,foldl,foldl1,foldr,foldr1,fromDouble,fromEnum,fromInt,fromInteger,fromIntegral,
		fromRational,fst,gcd,getChar,getContents,getLine,head,id,inRange,index,init,intToDigit,
		interact,ioError,isAlpha,isAlphaNum,isAscii,isControl,isDenormalized,isDigit,isHexDigit,
		isIEEE,isInfinite,isLower,isNaN,isNegativeZero,isOctDigit,isPrint,isSpace,isUpper,iterate,
		last,lcm,length,lex,lexDigits,lexLitChar,lines,log,logBase,lookup,map,mapM,mapM_,max,
		maxBound,posMax,negMax,maximum,maybe,min,minBound,minimum,mod,negate,not,notElem,null,numerator,odd,
		or,ord,pi,pred,primExitWith,print,product,properFraction,putChar,putStr,putStrLn,quot,
		quotRem,range,rangeSize,read,readDec,readFile,readFloat,readHex,readIO,readInt,readList,readLitChar,
		readLn,readOct,readParen,readSigned,reads,readsPrec,realToFrac,recip,rem,repeat,replicate,return,
		reverse,round,scaleFloat,scanl,scanl1,scanr,scanr1,seq,sequence,sequence_,show,showChar,showInt,
		showList,showLitChar,showParen,showSigned,showString,shows,showsPrec,significand,signum,sin,
		sinh,snd,span,splitAt,sqrt,subtract,succ,sum,tail,take,takeWhile,tan,tanh,threadToIOResult,toEnum,
		toInt,toInteger,toLower,toRational,toUpper,truncate,uncurry,undefined,unlines,until,unwords,unzip,
		unzip3,userError,words,writeFile,zip,zip3,zipWith,zipWith3,listArray,doParse,empty,for,initTo,
        assert,compose,checkGE,maxEvens,empty,create,get,set,initialize,idVec,fastFib,fibMemo,
        ex1,ex2,ex3,incr,inc,dec,isPos,positives,find,insert,len,size,union,fromList,initUpto,trim,
        insertSort,decsort,qsort,reverse,append,upperCase, ifM, whileM, get, decrM, diff,
        project, select, leq, elts, keys, dkeys, dfun, addKey, pTrue, emptyRD, rFalse,
        	dom, rng, isI, isD, isS, movie1, movie2,  toI, toS, toD, good_titles, runState, ret,
        	update, getCtr, setCtr, ctr, rdCtr, wrCtr, ifTest, whileTest, posCtr, zeroCtr, decr, decCtr,
        	pread , pwrite , plookup , pcontents, pcreateF , pcreateFP, pcreateD, active, caps, pset, eqP,
        	write, contents, alloc, derivP, copyP, createDir, store, copyRec, copySpec,
        	forM_, when, flookup, fread, createDir, pcreateFile, isFile, copyFrame, ?
	},
	emphstyle={[1]\color{haskellblue}},
	emph=
	{[2] 	Show,Eq,Ord,Num,UpClosed,Comp,Wit,Witness,Inductive,Meet,Flip,TRUE,Nat,Pos,Neg,IntGE,Plus,List,
        Bool,Char,Double,Either,Float,IO,Integer,Int,Maybe,
        Ordering,Rational,Ratio,ReadS,ShowS,String,Word8,
        InPacket,Tree,Vec,NullTerm,IncrList,DecrList,
        UniqList,BST,MinHeap,MaxHeap,World,RIO,IO,HIO,Post,Pre,
        Privilege, Prop, Chain, ChainTy, Range, Dict, RD, Dom, Set, P, Univ, Schema, MovieSchema, RT,
        TDom, TRange, MoviesTable, RTSubEqFlds, RTEqFlds, Disjoint, Union, Ret, Seq, Trans, Map,
        Pure, Then, Else, Exit, Inv, OneState, Priv, Path, FH, Stable,
				Prop, Nat,
	},
	emphstyle={[2]\color{castanho_ulisses}},
	emph=
	{[3]
		case,class,data,deriving,do,else,if,import,in,infixl,infixr,instance,let,
		module,of,primitive,then,refinement,type,where,forall,bound,and
		measure,reflect,predicate, assume
	},
	emphstyle={[3]\color{preto_ulisses}\textbf},
	emph=
	{[4]
		quot,rem,div,mod,elem,notElem,seq
	},
	emphstyle={[4]\color{castanho_ulisses}\textbf},
	emph=
	{[5]
		EQ,GT,LT,Left,Right
	},
	emphstyle={[5]\color{preto_ulisses}\textbf},
	emph=
	{[6]
	    axiomatize, measure, inline
	},
	emphstyle={[6]\color{lcolor}}
}

\lstnewenvironment{code}
{\lstset{language=HaskellUlisses}}
{}

\lstnewenvironment{mcode}
{\lstset{language=HaskellUlisses,columns=fullflexible,keepspaces,mathescape}}
{}

\lstMakeShortInline[language=HaskellUlisses,mathescape,keepspaces,mathescape,basicstyle=\ttfamily\small,breakatwhitespace]@

\lstdefinelanguage{Pseudo} {
	basicstyle=\ttfamily\codesize,
	sensitive=true,
  mathescape=true,
	morecomment=[l][\color{gray_ulisses}\ttfamily\codesize]{--},
	morecomment=[s][\color{gray_ulisses}\ttfamily\codesize]{\{-}{-\}},
	morestring=[b]",
	showstringspaces=false,
	numberstyle=\codesize,
	numberblanklines=true,
	showspaces=false,
	breaklines=true,
	showtabs=false
}
 
\title{Gradual Liquid Type Inference}
\titlenote{This work is partially funded by CONICYT, FONDECYT Regular Project 1150017, Chile, and is in part supported by the \href{https://erc.europa.eu/}{European Research Council} under \href{https://secure-compilation.github.io/}{ERC Starting Grant SECOMP (715753).}}

\author{Niki Vazou}
\affiliation{
  \institution{University of Maryland}
  \department{Dept.\ of Computer Science}
  \streetaddress{8223 Paint Branch Drive}
  \city{College Park}
  \state{Maryland}
  \postcode{20742}
  \country{USA}
}

\author{{\'E}ric Tanter}
\affiliation{
  \institution{University of Chile \& Inria Paris}
  \department{Computer Science Department (DCC)}
  \streetaddress{Beauchef 851}
  \city{Santiago}
  \country{Chile}
}

\author{David Van Horn}
\affiliation{
  \institution{University of Maryland}
  \department{Dept.\ of Computer Science}
  \streetaddress{8223 Paint Branch Drive}
  \city{College Park}
  \state{Maryland}
  \postcode{20742}
  \country{USA}
}

\begin{document}

\begin{abstract}
Refinement types allow for lightweight program verification by enriching types types with logical predicates. Liquid typing provides a decidable refinement inference mechanism that is convenient but subject to two major issues: 
(1) inference is global and requires top-level
annotations, making it unsuitable for inference of modular code
components and prohibiting its applicability to library code, and (2)
inference failure results in obscure error messages. These difficulties seriously hamper the migration of existing code to use refinements.

This paper shows that \emph{gradual liquid type inference}--a novel combination of liquid inference and gradual refinement types--addresses both issues.
Gradual refinement types, which support imprecise predicates that are optimistically interpreted, can be used in argument positions to constrain liquid inference so that the global inference process effectively infers modular specifications usable for library components.
Dually, when gradual refinements appear as the result of inference, they
signal an inconsistency in the use of static refinements.  
Because liquid refinements are drawn from a
finite set of predicates, in gradual liquid type inference
we can enumerate the \emph{safe concretizations} of each imprecise refinement, 
\ie the static refinements that justify why a program is gradually well-typed.
This enumeration is useful for static liquid type error explanation, since the safe concretizations exhibit all the potential inconsistencies that lead to
static type errors.

We develop the theory of gradual liquid type inference and explore its pragmatics in the setting of Liquid Haskell. 
To demonstrate the utility of our approach, we develop an
interactive tool, \toolname, for gradual liquid type inference in
Liquid Haskell that both infers modular types and explores safe
concretizations of gradual refinements. We report on the use of \toolname for error reporting
and discuss a case study on
the migration of three commonly-used Haskell list manipulation
libraries into Liquid Haskell.

 \end{abstract}

\maketitle           

\section{Introduction}\label{sec:intro}

Refinement types~\citep{Freeman91} allow for 
lightweight program verification
by decorating existing program types 
with logical predicates. 
For instance, the type @{x:Int | 0 < x}@
denotes strictly positive integer values
and can be used to validate {\em at compile time}
the absence of division-by-zero errors by refining the type of division:
\begin{code}
  (/) :: Int -> {x:Int | 0 < x} -> Int 
\end{code}

A major challenge with refinement types is to support decidable automatic checking and {\em inference}. To this end, {\em \mbox{liquid} types} restrict refinement predicates to decidable theories~\citep{LiquidPLDI08}.
The type @{x:Int | k}@ 
describes integer values refined with some predicate @k@, 
that is automatically solved based on unifying the constraints generated at each use of @x@, resulting in a concrete refinement drawn from 
a {\em finite} domain of template refinements.
The attractiveness of liquid types for programmers is usability.
Verification only requires specification of top-level functions, 
while all intermediate types can be automatically inferred
and checked. 

It has long been recognized that automatic type inference makes error reporting very challenging~\citep{Wand86}, resulting in severe usability problems. In particular, at an ill-typed function application, the system should decide whether to blame the function definition or the client; the particular choice will unfortunately and inevitably expose some of the internals of the inference procedure on to the user. 
The situation is even worse with liquid type inference, because refinement verification is more involved than standard Hindley-Milner style unification. 
The difficulty of understanding error messages from liquid inference in turn makes it really hard to progressively migrate non-refined code to refined code, \eg~from Haskell to Liquid Haskell~\citep{Vazou14}.

For instance, consider the following function @divIf@ that either inverts  its argument @x@ if it is positive, or inverts @1-x@ otherwise:
\begin{mcode}  
  divIf x = if isPos x then 1/x else 1/(1-x)
\end{mcode}
The function relies on an imported function @isPos@ of type @Int -> Bool@. 
If we want to give @divIf@ a refined type, liquid inference starts from the signature:
\begin{mcode}  
  divIf :: x:{ Int | k$_x$ } -> {o:Int | k$_o$ } 
\end{mcode}
and solves the predicate variables @k$_x$@ and @k$_o$@ based on the use sites. However, without any uses of @divIf@ in the considered source code, and for reasons we will clarify in due course, liquid inference infers the useless refinements:
\begin{mcode}  
  divIf :: x:{ Int | false } -> {o:Int | false } 
\end{mcode}
The false refinement on the argument means that @divIf@ is dead code, 
since liquid type inference relies on a {\em closed world assumption}; and 
@divIf@ has, at inference time, no clients.  

Conversely, if the code base at inference time includes a ``positive'' client call @divIf 1@, the inferred precondition would be @0 < x@, triggering a type error in the {\em definition}
of @divIf@ due to the lack of information about @isPos@.
This hard-to-predict and moving blame is not unique 
to liquid type inference; it frequently appears in 
type inference engines, yielding hard-to-debug error messages.

A key contribution of this paper is to recognize that, in such situations, treating the inferred refinement as an {\em unknown} refinement --- in the sense of gradual typing~\citep{siekTaha:sfp2006} --- allows us to 
eliminate the closed world assumption, as well as globally
explain liquid type errors. 
Specifically, we adapt the gradual refinement types of \citet{Lehmann17} to the setting of liquid type inference, yielding gradual liquid type inference. 
Programmers can introduce gradual refinements such as @{x: Int | ?}@
and inference exhaustively searches for 
\textit{safe concretizations} (\scshorts for short), \ie the concrete refinements that can
replace each occurrence of a gradually-refined variable to make the program well-typed. 
These \scshorts
can then be used to understand liquid type errors and assist in migrating programs to adopt liquid types.

\paragraph{Contributions} This paper makes the following contributions:
\begin{itemize}[leftmargin=*]
\item We give a semantics and inference algorithm for gradual liquid types (\S~\ref{sec:theory}), by exploiting the abstract interpretation approach to gradual language design~\citep{adt}.  We prove that inference is correct and satisfies the static criteria for gradual languages~\citep{siek15}. The latter result is, to the best of our knowledge, novel for a gradual inference system.
\item We implement gradual liquid type inference in \toolname,
as an extension of \liquidHaskell (\S~\ref{sec:implementation}). The implementation integrates a number of optimizations and heuristics in order to be applicable to existing Haskell libraries.
\item \toolname takes as input a Haskell
program annotated with gradual refinements
and generates an interactive program
that presents all valid static choices to interpret 
each separate occurrence of an unknown refinement, if any.
The user can explore suggested predicates and decide which ones to use
so that their program is well-typed (\S~\ref{sec:error-explanation}).

\item We use \toolname for user-guided migration of three
existing Haskell libraries (1260 LoC) to \liquidHaskell (\S~\ref{sec:migration}),
demonstrating that
gradual liquid type inference can feasibly be supported and used 
interactively for program migration.
\end{itemize}

The essence of the novelty of our approach---that gradual inference based on abstract interpretation can be fruitfully used for error reporting and program migration---is not restricted to refinement types. 
We conjecture that the principles of gradual types can be used for 
type error explanation in further typing systems. 
This paper proceeds as follows.
We first provide an informal overview of liquid type inference,
gradual refinement types, and gradual liquid type inference
(\S~\ref{sec:overview}).  Next, we formally establish the necessary
background on liquid types and gradual refinements
(\S~\ref{sec:background}), before presenting the main results
(\S~\ref{sec:theory}--\ref{sec:migration}).  Finally, we discuss
related work (\S~\ref{sec:related}) and conclude
(\S~\ref{sec:conclusion}). Auxiliary
definitions as well as proofs of theorems can be found in Appendix~\ref{appendix:section:metatheory}. 
Source code for \toolname is available at:
\begin{center}
\url{https://github.com/ucsd-progsys/liquidhaskell/tree/gradual}
\end{center}

\section{Background and Overview}\label{sec:overview}

We start by recalling the basic notions of refinement types (\S~\ref{subsec:overview:refinementtypes}), 
their decidable fragment known as liquid types (\S~\ref{subsec:overview:liquid}), 
and recent work on gradual refinements (\S~\ref{subsec:overview:gradual}). 
Then (\S~\ref{subsec:overview:glt}) we combine gradual and liquid types 
to build a usable interactive inference procedure that
can be used for both error explanation and program migration.

\subsection{Refinement Types}\label{subsec:overview:refinementtypes}

To explain refinement type checking, assume the below type signatures for the example of~\S~\ref{sec:intro}:
\begin{mcode}
  isPos :: x:Int -> {b:Bool | b $\Leftrightarrow$ 0 < x}
  divIf :: Int -> Int 
\end{mcode}
With these specifications, @divIf@ is well-typed because the 
positive restriction in the precondition of @(/)@ is provably satisfied.
Refinement type checking proceeds in three steps, described below.

\paragraph{Step 1: Constraint Generation}
Based on the code and the specifications, refinement subtyping constraints are generated; for our example, two subtyping constraints are generated, 
one for each call to @(/)@.
They stipulate that, in the environment with
the argument @x@ and the boolean branching guard @b@, 
the second argument to @(/)@ (\ie @v = x@ and @v = 1-x@, \textit{resp.})
is {\em safe}, \ie it respects the precondition @0 < v@.
\begin{mcode}	
  b:{b $\Leftrightarrow$ 0 < x $\land$ b}   ${\vdash}$ {v | v = x}   $\preceq$ {v | 0 < v}	
  b:{b $\Leftrightarrow$ 0 < x $\land$ not b} ${\vdash}$ {v | v = 1-x} $\preceq$ {v | 0 < v}	
\end{mcode}
For space, we write @{v | p}@ to denote @{v:t | p}@
when the type @t@ is clear; we omit the refinement variables from the environment, simplifying @x:{x | p}@ to @x:{p}@, and we skip uninformative 
refinements such as @x:{true}@.

In both constraints the branching guard @b@
is refined with the result refinement of @isPos@: 
@b $\Leftrightarrow$ 0 < x@. 
Also, the environment is strengthened with the value of the condition in each branch: @b@ in the @then@ branch,@not b@ in the @else@ branch.

\paragraph{Step 2: Verification Conditions}
Each subtyping constraint is reduced to a logical verification condition (\textit{VC}), 
that validates if,
assuming all the refinements in the environment, 
the refinement on the left-hand side implies the one on the right-hand side. 
For instance, the two constraints above reduce to the following VCs.
\begin{mcode}	
  b $\Leftrightarrow$ 0 < x $\land$  b $\Rightarrow$ v = x   $\Rightarrow$ 0 < v 	
  b $\Leftrightarrow$ 0 < x $\land$ not b $\Rightarrow$ v = 1-x $\Rightarrow$ 0 < v 	
\end{mcode}
 
\paragraph{Step 3: Implication Checking}
Finally, an SMT solver is used to check the 
validity of the generated VCs, 
and thus determine if the program is well-typed. 
Here, the SMT solver determines that both VCs are valid, 
thus @divIf@ is well-typed.

\subsection{Liquid Types}\label{subsec:overview:liquid}
When not all refinement types are spelled out explicitly, one can use {\em inference}. For instance, the well-typedness of @divIf@ crucially relies on the guard predicate, as propagated by the refinement type 
of @isPos@. We now explain how liquid typing~\citep{LiquidPLDI08} 
infers a type for @divIf@ in case @isPos@ is an imported, 
{\em unrefined} function. 
Liquid types introduce refinement type variables, known as {\em liquid variables}, for unspecified refinements. 
As in \S~\ref{sec:intro}, the type of @divIf@ 
is assigned the liquid variables @k$_x$@ and @k$_o$@ for the 
input and the output refinements, \resp:
\begin{mcode}  
   divIf :: x:{ Int | k$_x$ } -> {o:Int | k$_o$ } 
\end{mcode}
After generating subtyping constraints as in~\S~\ref{subsec:overview:refinementtypes}, the inference procedure attempts to find a solution for the liquid variables such that all the constraints are satisfied. 
If no solution can be found, the program is deemed ill-typed.

\paragraph{Step 1: Constraint Generation}
After introduction of the liquid variables, the following subtyping constraints are generated for @divIf@.
\begin{mcode}	
  x:{k$_x$}, b:{b}  ${\vdash}$ {v | v = x  } $\preceq$ {v | 0 < v}	
  x:{k$_x$}, b:{not b} ${\vdash}$ {v | v = 1-x} $\preceq$ {v | 0 < v}	
\end{mcode}

\paragraph{Step 2: Constraint Solving}
Liquid inference then solves the liquid variables
@k@ so that the subtyping constraints are satisfied. 
The solving procedure takes as input 
a {\em finite} set of refinement {\em templates} $\quals^\star$
abstracted over program variables.
For example, the template set $\quals^\star$ below describes ordering predicates, 
with $\star$ ranging over program variables.
\begin{mcode}
  $\quals^\star$ = {$0<\star$, $0\leq\star$, $\star<0$, $\star\leq 0$, $\star < \star$, $\star \leq \star$}
\end{mcode} 

Next, for each liquid variable, 
the set $\quals^\star$ is instantiated
with all program variables in scope, to generate well-sorted predicates.
Instantiation of $\quals^\star$ 
for the liquid variables @k$_x$@ and @k$_o$@
leads to the following concrete predicate candidates. 
\begin{mcode}
  $\quals^x$ = { $0\!<\!x$, $0\!\leq\!x$, $x\!<\!0$, $x\!\leq\!0$ }
  $\quals^o$ = { $0\!<\!o$, $0\!\leq\!o$, $o\!<\!0$, $o\!\leq\!0$, $o\! <\! x$, $x\!<\!o$, $\dots$ }
\end{mcode} 

Finally, inference iteratively computes the strongest solution 
for each liquid variable that satisfies the constraints. 
It starts from an initial solution 
that maps each variable to the logical conjunction of all the 
instantiated templates
@$\ltsol$ = {k$_x \mapsto \bigwedge \quals^x$, k$_o \mapsto \bigwedge \quals^o$}@.
It then repeatedly filters out predicates of the solution until all constraints are satisfied. 

In our example, the initial solution includes the predicates $0<x$ and
$x<0$ is \emph{contradictory}, the conjunction of solutions implies
false, so both liquid variables are solved to false.  As discussed
in~\S~\ref{sec:intro}, this inferred as false, solution is based on a
closed world assumption and is practically useless.  With client code
that imposes additional constraints, the inferred solution can be more
useful, though the reported errors can be hard to interpret.

\subsection{Gradual Refinement Types}\label{subsec:overview:gradual}
We observe that we can exploit gradual typing in order to assist inference and provide better support for error explanation and program migration. 
Instead of interpreting an unspecified refinement as a liquid variable, let us use the unknown gradual refinement @{ Int | ? }@ for the
argument type of @divIf@~\citep{Lehmann17}.
\begin{mcode}  
   divIf :: x:{ Int | ? } -> Int 
\end{mcode}
This gradual precondition specifies that 
for each {\em usage occurrence} of the argument @x@, 
there must \textit{exist} a concrete refinement (which we call a {\em safe concretization}, \scshort)
for which the (non-gradual) program type checks. Key to this definition is that the refinement that exists need not be unique to all occurrences of the identifier. 
Gradual refinement type checking proceeds as follows.

\paragraph{Step 1: Constraint Generation}
First, we generate the subtyping constraints derived from the definition of 
@divIf@ that now contain gradual refinements.
\begin{mcode}	
  x:{?}, b:{b}  ${\vdash}$ {v | v = x  } $\preceq$ {v | 0 < v}
  x:{?}, b:{not b} ${\vdash}$ {v | v = 1-x} $\preceq$ {v | 0 < v}	
\end{mcode}

\paragraph{Step 2: Gradual Verification Conditions}
Each subtyping reduces to a VC, 
where each gradual refinement such as @x:{?}@ translates intuitively 
to an existential refinement (@$\exists$ p. p x@).
The solution of these existentials are 
the safe concretizations (\scshorts) of the program.
Here, we informally use @$\exists^{\egrad}$ p@
to denote such existentials over predicates, and call the corresponding verification conditions {\em gradual VCs} (GVCs).
For example, the GVCs for @divIf@ are the following.
\begin{mcode}	
  ($\exists^{\egrad}$ p$_{\text{then}}$. p$_{\text{then}}$ x) $\land$ b  $\Rightarrow$ v=x   $\Rightarrow$ 0 < v	
  ($\exists^{\egrad}$ p$_{\text{else}}$. p$_{\text{else}}$ x) $\land$ not b $\Rightarrow$ v=1-x $\Rightarrow$ 0 < v	
\end{mcode}

\paragraph{Step 3: Gradual Implication Checking}
Checking the validity of the generated GVCs
is an open problem. 
In the @divIf@ example, 
we can, by observation, find the \scshorts
that render the GVCs valid: @p$_{\text{then}}$ x $\mapsto$ 0 < x@ and 
@p$_{\text{else}}$ x $\mapsto$ x $\leq$ 0@.

More importantly, we can present these solutions to the user 
as the conditions under which @divIf@ is well-typed.
We use the \scshorts to explain to the user that 
for @divIf@ to type check under a static type, 
the gradual precondition @?@ should be replaced with a refinement 
that implies @0 < x@ in the @then@ branch and @x $\leq$ 0@ in the @else@ 
branch.
Since such a refinement cannot exist, we use \scshorts to explain to the user 
that the @?@ cannot be replaced in the type of @divIf@, 
unless the user modifies the code 
(here, strengthen the postcondition of @isPos@, as in~\S~\ref{subsec:overview:refinementtypes}).

Our goal is to find an algorithmic procedure 
that solves GVCs. 
\citet{Lehmann17} describe how GVCs over linear arithmetic 
can be checked, while 
\citet{Courcelle12} describe a more general logical fragment 
(monadic second order logic) with an algorithmic decision procedure. 
Yet, in both cases we lose the justifications, \ie~the \scshorts,  
and thus the opportunity to use such \scshorts for error explanation and migration assistance.

\subsection{Gradual Liquid Type Inference}\label{subsec:overview:glt}

To algorithmically solve GVCs we can exhaustively search for \scshorts in 
the {\em finite} predicate domain of liquid types.
In between constraint generation and constraint solving, 
gradual liquid type inference concretizes the constraints
by instantiating gradual refinements with each possible 
liquid template.
 
\paragraph{Step 1: Constraint Generation} 
Constraint generation is performed exactly like gradual refinement types, 
leading to the constraints of~\S~\ref{subsec:overview:gradual}
for the @divIf@ example.

\paragraph{Step 2: Constraint Concretization}
We exhaustively generate all the possible concretizations 
of the constraints. 
For example, @x:{?}@ can be concretized to any predicate from the $\quals^x$ set of~\S~\ref{subsec:overview:liquid}, 
yielding $|\quals^x|^2 = 16$ concrete constraint sets, among which the following two:

\begin{enumerate}[leftmargin=*]
\item\label{concrete:invalid}
\textbf{Concretization for} @p$_{\text{then}}$ x $\mapsto$ 0 < x@, @p$_{\text{else}}$ x $\mapsto$ 0 < x@:

\begin{mcode}	
 x:{0<x}, b:{b}  ${\vdash}$ {v | v = x  } $\preceq$ {v | 0 < v}
 x:{0<x}, b:{not b} ${\vdash}$ {v | v = 1-x} $\preceq$ {v | 0 < v}	
\end{mcode}

\item\label{concrete:valid}
\textbf{Concretization for} @p$_{\text{then}}$ x $\mapsto$ 0 < x@, @p$_{\text{else}}$ x $\mapsto x \leq 0$@:
\begin{mcode}	
 x:{0<x},  b:{b}  ${\vdash}$ {v | v = x  } $\preceq$ {v | 0 < v}
 x:{x$\leq$0}, b:{not b} ${\vdash}$ {v | v = 1-x} $\preceq$ {v | 0 < v}	
\end{mcode}
\end{enumerate}

\paragraph{Step 3: Constraint Solving}
After concretization, liquid constraint solving finds out the valid ones. 
In our example, the constraint~\ref{concrete:invalid} above is invalid
while~\ref{concrete:valid} is valid. 
Out of the 16 concrete constraints, only two are valid, 
with @p$_{\text{then}}$ x $\mapsto$ 0 < x@ and
@p$_{\text{else}}$ x $\mapsto x \leq 0$@
or @p$_{\text{else}}$ x $\mapsto x < 0$@.
Thus, @divIf@ type checks and also the inference 
provides to the user the \scshorts 
of @p$_{\text{then}}$@ and  @p$_{\text{else}}$@
as an explanation of type checking.

\paragraph{Application to Error Explanation}
The contradictory solutions of the gradual refinement of the argument 
indicates to the user that the code cannot statically type check. The user needs to  edit either the code or the refinement types.
For example,
the code can be fixed by providing a precise type for @isPos@ 
(following~\S~\ref{subsec:overview:refinementtypes})
or by restricting @divIf@'s domain to positive numbers 
(rendering the @else@ branch as dead code). 

Unlike current liquid type inference, 
our algorithm does not assume a closed world, since the \scshorts
it provides do not depend on call sites of the functions. 
Thus, the output of the algorithm is modular: 
it explains the code contradictions that lead to type errors
but these contradictions are generated in the function definition 
and do not rely on the arguments of function calls. 

Importantly, gradual liquid types are used for error explanation 
using the language of contradictions in refinement predicates that the user understands. 
In this example, the output of our algorithm 
informs the user that the program cannot type check because 
the same refinement should be solved to 
@0 < x@ in the @then@ branch and to 
@x < 0@ in the @else@ branch, which is impossible.  
This explanation is much more informative 
than the current liquid type algorithm, which would, 
generate a type error either in the definition of 
@divIf@ or at its call sites, following the closed world assumption,
as discussed in~\S~\ref{sec:intro}.

In~\S~\ref{sec:error-explanation} we use gradual liquid types
to explain to the user the infamous off-by-one bug. 
Since the algorithm is exhaustively searching all potential solutions, 
it is exponentially slow on the number of potential solutions. 
Yet, in~\S~\ref{sec:migration} we show that we can apply our technique on 
real Haskell code due to the three following reasons: 
\begin{itemize}[leftmargin=*]
\item Our algorithm is a modular, per-function analysis. Thus
inference time depends on the size of the analyzed function, and not 
on the size of the whole codebase to be type checked. 
\item Our implementation is user interactive, thus
the exponential complexity is not a problem in practice,
since our algorithm runs in the background while expecting the user input.  
Once a \scshort is found, the system presents it to the user, 
while looking for the next \scshorts in the background. 
\item Finally, in~\S~\ref{sec:implementation} we discuss further
technical optimizations that make our theoretically-exponential
algorithm tractable and usable in practice.   
\end{itemize}

In~\S~\ref{sec:theory} we formalize the inference steps 
and prove the correctness and the gradual criteria 
of our algorithm. Various implementation considerations necessary for the algorithm to scale are described in \S~\ref{sec:implementation}. We report on its use for error explanation and program migration in \S~\ref{sec:error-explanation} and \S~\ref{sec:migration}.

\section{Liquid Types and Gradual Refinements}\label{sec:background}

We briefly provide the technical background 
required to describe gradual liquid types. 
We start with the semantics and rules of a generic refinement type system 
(\S~\ref{subsec:refinementtypes}) which we then adjust to describe both
liquid types (\S~\ref{subsec:liquidtypes}) 
and gradual refinement types (\S~\ref{subsec:gradualtypes}).

\subsection{Refinement Types}\label{subsec:refinementtypes}
\begin{figure}[t!]
\centering
\small
\begin{minipage}{0.4\textwidth}
$
\begin{array}{rrcl}
\emphbf{Constants} \quad
  & c
  & ::=
  & \land \spmid \lnot \spmid = \spmid \dots \\
  && \spmid & \etrue \spmid \efalse \\
  && \spmid & 0, 1,-1, \dots
\\[0.03in]

\emphbf{Values} \quad
  & \vparam & ::=&  c
  \spmid \efun{x}{\typ}{\expr}
\\[0.03in]

\emphbf{Expressions} \quad
  & \expr & ::=& \vparam \spmid x  \spmid \eapp{\expr}{x} \\
  &   & \spmid & \eif{x}{\expr}{\expr} \\
  &   & \spmid & \elet{x}{\typ}{\expr}{\expr}\\
  &   & \spmid & \elettyp{x}{\typ}{\expr}{\expr}
\\[0.03in]

\emphbf{Predicates} \quad
  & \pred
  & ::=
  & \eparam
\end{array}
$
\end{minipage}\hspace{1cm}
\begin{minipage}{0.4\textwidth}
$
\begin{array}{rrcl}
\emphbf{Basic Types} \quad
  & \btyp
  & ::=
  & \tint \spmid \tbool
\\[0.03in]

\emphbf{Types} \quad
  & \tparam
  & ::= &   \tref{x}{\btyp}{\pred} 
            \spmid \tfun{x}{\tparam}{\tparam} 
\\[0.05in]

\emphbf{Environment} \quad
  & \env
  & ::=  & \cdot \spmid \env, \bind{x}{\typ}
\\[0.05in]

\emphbf{Substitution} \quad
  & \esubst
  & ::=  & \cdot \spmid \esubst, \ebind{x}{\expr}
\\[0.05in]

\emphbf{Constraint} \quad
  & C
  & ::=  & \iswellformed{\rparam}{\envparam}{\tref{v}{\btyp}{\pred} } \\
  & & \spmid & \issubtype{\rparam}{\envparam}{\tref{v}{\btyp}{\pred}}{\tref{v}{\btyp}{\pred}}
\end{array}
$
\end{minipage}
\vspace{-3mm}
\caption{Syntax of \reflang.}
\vspace{-3mm}
\label{fig:syntax}
\end{figure}
\begin{figure*}[t]
\small
\judgementHead{Typing}{\hastype{\rparam}{\envparam}{\eparam}{\tparam}}\\
$$
\inference{
	\envparam (x) = \tref{v}{\btyp}{\_}
}{
	\hastype{\rparam}{\envparam}{x}{\tref{v}{\btyp}{\refparam{\rparam}{v=x}}}
}[\ruletvarbase]
\qquad
\inference{
	\envparam (x)\ \text{is a function type}
}{
	\hastype{\rparam}{\envparam}{x}{\envparam (x)}
}[\ruletvar]
\qquad
\inference{
}{
	\hastype{\rparam}{\envparam}{c}{\tc{c}}
}[\ruletconst]
$$

$$
\inference{
	\hastype{\rparam}{\envparam}{\expr}{\typ_e}
	&& 
	\iswellformed{\rparam}{\envparam}{\typ}
	&&
	\issubtype{\rparam}{\envparam}{\tparam_e}{\tparam}
}{
	\hastype{\rparam}{\envparam}{\expr}{\typ}
}[\ruletsub]
\qquad
\inference{
	\hastype{\rparam}{\envparam,\bind{x}{\inferred{\tparam_x}}}{\eparam}{\tparam} &&
	\iswellformed{\rparam}{\envparam}{\tfun{x}{\inferred{\tparam_x}}{\tparam}}
}{
	\hastype{\rparam}{\envparam}{\efun{x}{\typ_x}{\eparam}}{(\tfun{x}{\inferred{\tparam_x}}{\tparam})}
}[\ruletfun]
$$

$$
\inference{
	\hastype{\rparam}{\envparam}{y}{\tparam_x}
	\\ 
	\hastype{\rparam}{\envparam}{\eparam}{(\tfun{x}{\tparam_x}{\tparam})}
}{
	\hastype{\rparam}{\envparam}{\eapp{\eparam}{y}}{\tparam\sub{x}{y}}
}[\ruletapp]
\qquad
\inference{
    \text{fresh}\ x'
    &&
    \envparam_1 \defeq \envparam, \bind{x'}{\tref{v}{\tbool}{\refparam{\rparam}{x}}}
    &&
    \envparam_2 \defeq \envparam, \bind{x'}{\tref{v}{\tbool}{\refparam{\rparam}{\lnot x}}}
    \\
	\hastype{\rparam}{\envparam}{x}{\tref{v}{\tbool}{\_}}
	&& 
	\hastype{\rparam}{\envparam_1}{\eparam_1}{\inferred{\typ}}
	&& 
	\hastype{\rparam}{\envparam_2}{\eparam_2}{\inferred{\typ}}
	&&
	\iswellformed{\rparam}{\envparam}{\inferred{\typ}}
}{
	\hastype{\rparam}{\envparam}{\eif{x}{\eparam_1}{\eparam_2}}{\inferred{\typ}}
}[\ruletif]
$$

$$
\inference{
	\hastype{\rparam}{\envparam}{\eparam_x}{\tparam_x}
	&&
	\hastype{\rparam}{\envparam,\bind{x}{\tparam_x}}{\eparam}{\inferred{\typ}}
    &&
	\iswellformed{\rparam}{\envparam}{\inferred{\typ}}
}{
	\hastype{\rparam}{\envparam}{\elet{x}{\typ_x}{\eparam_x}{\eparam}}{\inferred{\typ}}
}[\ruletlet]
\qquad
\inference{
	\hastype{\rparam}{\envparam}{\eparam_x}{\tparam_x}
	&&
	\hastype{\rparam}{\envparam,\bind{x}{\tparam_x}}{\eparam}{\inferred{\typ}}
	&&
	\iswellformed{\rparam}{\envparam}{\inferred{\typ}}
	&&
	\iswellformed{\rparam}{\envparam}{\tparam_x}
}{
	\hastype{\rparam}{\envparam}{\elettyp{x}{\tparam_x}{\eparam_x}{\eparam}}{\inferred{\typ}}
}[\ruletspec]
$$

\judgementHead{Sub-Typing}{\issubtype{\rparam}{\envparam}{\tparam}{\tparam}}\\
$$
\inference{
	\valid{\issubtype{\rparam}{\envparam}{\tref{v}{\btyp}{\pred_1}}{\tref{v}{\btyp}{\pred_2}}}
}{
	\issubtype{\rparam}{\envparam}{\tref{v}{\btyp}{\pred_1}}{\tref{v}{\btyp}{\pred_2}}
}[\rulesbase]
\qquad
\inference{
	\issubtype{\rparam}{\envparam}{\tparam_{x2}}{\tparam_{x1}}
	&&
	\issubtype{\rparam}{\envparam, \bind{x}{\tparam_{x2}} }{\tparam_1}{\tparam_2}
}{
	\issubtype{\rparam}{\envparam}{\tfun{x}{\tparam_{x1}}{\tparam_1}}{\tfun{x}{\tparam_{x2}}{\tparam_2}}
}[\rulesfun]
$$

\judgementHead{Well-Formedness}{\iswellformed{\rparam}{\envparam}{\tparam}}\\
$$
\inference{
	\valid{\iswellformed{\rparam}{\envparam}{\tref{v}{\btyp}{\pred}}}
}{
	\iswellformed{\rparam}{\envparam}{\tref{v}{\btyp}{\pred}}
}[\rulewbase]
\qquad
\inference{
	\iswellformed{\rparam}{\envparam}{\tparam_x}
	&&
	\iswellformed{\rparam}{\envparam, \bind{x}{\tparam_x}}{\tparam}
}{
	\iswellformed{\rparam}{\envparam}{\tfun{x}{\tparam_x}{\tparam}}
}[\rulewfun]
$$
\caption{Static Semantics of \reflang. (Types colored in blue need to be inferred.)}
\label{fig:rules}
\end{figure*} 
\paragraph{Syntax}
Figure~\ref{fig:syntax} presents the syntax of 
a standard functional language with refinement types, \reflang. 
The expressions of the language include constants,
lambda terms, variables, function applications,
conditionals, and let bindings.
Note that the argument of a function application needs to be syntactically a variable, as must the condition of a conditional; this normalization is standard in refinement types as it simplifies the formalization~\citep{LiquidPLDI08}. There are two let binding forms, one where the type of the bound variable is inferred and one where it is explicitly declared.

\reflang types include {\em base refinements} $\tref{\rbind}{\btyp}{\pred}$ where 
$\btyp$ is a base type (\tint or \tbool) refined with the logical predicate $\pred$. A predicate can be any 
expression $\expr$, which can refer to $\rbind$. 
Types also include dependent function types 
$\tfun{x}{\typ_x}{\typ}$, where $x$ is bound to the function argument 
and can appear in the result type $\typ$.
As usual, we write $\btyp$ as a shortcut for $\tref{\rbind}{\btyp}{\etrue}$ and $\typ_x \rightarrow \typ$ as a shortcut for $\tfun{x}{\typ_x}{\typ}$
when $x$ does not appear in $\typ$.

\paragraph{Denotations}
Following~\citet{Knowles10},
each type of \reflang denotes a set of expressions. 
The denotation of a base refinement includes all 
expressions that either diverge or evaluate to base values that satisfy the associated predicate.
We write $\totrue{e}$ to represent that $e$ is (operationally) valid
and $\term{e}$ to represent that $e$ terminates:
\[ 
\begin{array}{ccc}
\totrue{e} \defeq \evals{e}{\etrue} & \quad &
\term{e} \defeq \exists v. \evals{e}{v} 
\end{array}
\]
where \evals{\cdot}{\cdot} is the reflexive, transitive closure of the small-step reduction relation. Denotations are naturally extended to function types and environments (as sets of substitutions).
\[ \arraycolsep=0.5pt %
\begin{array}{rcl}
\embed{\tref{\rbind}{\btyp}{\pred}} &\defeq& \{e \spmid \vdash e :\btyp,\text{if}\  \term{e}\ \text{then}\ \totrue{\pred\sub{\rbind}{e}} \}\\
\embed{\tfun{x}{\tparam_x}{\tparam}} &\defeq& \{\eparam \spmid \forall \eparam_x\in\embed{\tparam_x}. \eapp{\eparam}{\eparam_x}\in\embed{\tparam\sub{x}{\eparam_x}} \} \\
\embed{\envparam} &\defeq& \{\esubst \spmid \forall \bind{x}{\tparam}\in\envparam. \ebind{x}{\eparam} \in \esubst \land \eparam \in \embed{\applysub{\tparam}{\esubst}} \}
\end{array}
\]

\paragraph{Static semantics}
Figure~\ref{fig:rules} summarizes the standard typing
rules that characterize whether an expression belongs to the denotation of a type~\citep{LiquidPLDI08,Knowles10}. 
Namely, $\expr\in\embed{\typ}$ \textit{iff} $\hastype{}{}{\expr}{\typ}$.
We define three kinds of relations 
1.~typing, 
2.~subtyping, and 
3.~well-formedness.

\begin{enumerate}[leftmargin=*]
\item\label{rules:typechecking}
\textit{Typing:}
\hastype{\rparam}{\env}{\expr}{\typ}
\textit{iff} 
$\forall \esubst \in \embed{\env}. 
\applysub{\expr}{\esubst} \in \embed{\applysub{\typ}{\esubst}}
$.\\
Rule \ruletvarbase refines the type of a variable with its exact value.
Rule \ruletconst types a constant $c$ using the function $\tc{c}$ 
that is assumed to be sound, \ie we assume that for each constant 
$c$, $c\in\embed{\tc{c}}$. 
Rule \ruletsub allows to weaken the type of a given expression by subtyping, discussed below.
Rule \ruletif achieves path sensitivity by typing 
each branch under an environment strengthened with the value of the condition.
Finally, the two let binding rules \ruletlet and \ruletspec only differ in whether the type of the bound variable is inferred or taken from the syntax. Note that the last premise, a well-formedness condition, ensures that the bound variable does not escape (at the type level) the scope of the let form.

\item\label{rules:subtyping}
\textit{Subtyping:}
\issubtype{}{\env}{\typ_1}{\typ_2}
\textit{iff} 
$\forall\esubst\!\in\!\embed{\env},\expr\!\in\!\embed{\applysub{\typ_1}{\esubst}}$.
$\expr\!\in\!\embed{\applysub{\typ_2}{\esubst}}
$.\\
Rule \rulesbase uses the relation \valid{\cdot}
to check subtyping on basic types; we leave this relation abstract for now since we will refine it in the course of this section.
\citet{Knowles10}
define subtyping between base refinements as:
$$
\begin{array}{c}
\valid{\issubtype{\rparam}{\envparam}{\tref{\rbind}{\btyp}{\pred_1}}{\tref{\rbind}{\btyp}{\pred_2}}}\qquad
\textit{iff}\qquad
    \forall \esubst\in\embed{\envparam, \bind{\rbind}{\btyp}}. 
    \text{if}\ \totrue{\applysub{\pred_1}{\esubst}} \ \text{then}\ \totrue{\applysub{\pred_2}{\esubst}}
\end{array}
$$
This definition makes checking undecidable, as it quantifies over all substitutions. We come back to decidability below.

\item\label{rules:wellformedness}
\textit{Well-Formedness:}
Rule \rulewbase overloads %
\valid{\cdot}
to refer to well-formedness on base refinements. 
A base refinement 
$\tref{\rbind}{\btyp}{\pred}$ is well-formed only when 
$\pred$ is typed as a boolean:

$$
\begin{array}{c}
\valid{\iswellformed{}{\env}{\tref{\rbind}{\btyp}{\pred}}}
\qquad\textit{iff}\qquad
\hastype{}{\env,\bind{\rbind}{\btyp}}{\pred}{\tbool}
\end{array}
$$
\end{enumerate}

\paragraph{Inference}
In addition to being undecidable, the typing rules in Figure~\ref{fig:rules} are not syntax directed: several types do not come from the syntax of the program, but have to be guessed---they are colored in blue in Figure~\ref{fig:rules}. These are: the argument type of a function (Rule \ruletfun), the common (least upper bound) type of the branches of a conditional (Rule \ruletif), and the resulting type of let expressions (Rules \ruletlet and \ruletspec), which needs to be weakened to not refer to variable $x$ in order to be well-formed. 
Thus, to turn the typing relation into a type checking algorithm, one needs to address both decidability of subtyping judgments and inference of the aforementioned types.

\subsection{Liquid Types}\label{subsec:liquidtypes}

Liquid types~\citep{LiquidPLDI08} provide a decidable and efficient inference algorithm for the typing relation of Figure~\ref{fig:rules}. 
For decidability, the key idea of liquid types is to restrict refinement predicates to be drawn
from a \textit{finite} set of {\em predefined}, SMT-decidable predicates 
$\qual\in\iquals$.

\paragraph{Syntax}
The syntax of {\em liquid predicates}, written \lpred, is:
$$
\begin{array}{rcll}
  \lpred
  & ::=     & \etrue                        & \text{True}\\
  &\spmid  & \qual                         & \text{Predicate, with}\ \qual\in\iquals \\
  &\spmid  & \lpred \land \lpred           & \text{Conjunction}\\
  & \spmid & \lvar                         & \text{Liquid Variable}\\
[0.03in]

  \ltsol
  & ::=    & \cdot \spmid \ltsol, \kvar \mapsto \overline{q}  & Solution
\end{array}
$$
A liquid predicate can be 
true (\etrue), 
an element from the predefined set of predicates (\qual), 
a conjunction of predicates ($\lpred \land \lpred$), or 
a predicate variable (\lvar), called a {\em liquid variable}.
A {\em solution} \ltsol is a mapping from 
liquid variables to a set of elements of \iquals.
The set $\overline{\qual}$ represents a variable-free liquid predicate 
using \etrue for the empty set and conjunction to combine the elements otherwise.

\paragraph{Checking}
When all the predicates in \iquals belong to SMT-decidable theories,
validity checking of \reflang:
$\valid{\issubtype{}{\env}{\tref{\rbind}{\btyp}{\pred_1}}{\tref{\rbind}{\btyp}{\pred_2}}}
$
which quantifies over all embeddings of the typing environment, can be SMT automated in a sound and complete way. 
Concretely, a subtyping judgment  
\issubtype{}{\env}{\tref{\rbind}{\btyp}{\lpred_1}}{\tref{\rbind}{\btyp}{\lpred_2}}
is valid \textit{iff}
under all the assumptions of $\env$, 
the predicate $\lpred_1$ implies the predicate $\lpred_2$.
$$
\begin{array}{c}
\valid{\issubtype{}{\env}{\tref{\rbind}{\btyp}{\lpred_1}}{\tref{\rbind}{\btyp}{\lpred_2}}}\\
\textit{iff}\\
    \smtvalid{
    \bigwedge \{\lpred \mid \bind{x}{\tref{\rbind}{\btyp}{\lpred}} \in \env \}
    \Rightarrow \lpred_1
    \Rightarrow \lpred_2
    }
\end{array}
$$

\paragraph{Inference}
The liquid inference algorithm, defined in Figure~\ref{fig:infer}, first 
applies the rules of Figure~\ref{fig:rules} 
using liquid variables as the refinements
of the types that need to be inferred 
and then uses an iterative algorithm to solve 
the liquid variable as a subset of \iquals~\citep{LiquidPLDI08} 
(steps 1 and 2 of~\S~\ref{subsec:overview:liquid}). 

More precisely, given a typing environment $\ltenv$, 
an expression $\ltexpr$, and the fixed set of predicates 
$\iquals$, the function $\infer{\ltenv}{\ltexpr}{\iquals}$
returns the type of the expression $\ltexpr$ under the environment $\ltenv$, 
if it exists, or nothing otherwise.
It first generates a template type $\lttyp$ and a set of constraints @C@
that contain liquid variables in the types to be inferred. 
Then it generates a solution $\ltsol$ that satisfies all the constraints in @C@.
Finally, it returns the type $\lttyp$
in which all the liquid variables have been substituted by concrete refinements from the mapping in $\ltsol$.

\begin{figure}[t]
\begin{mcode}
  Infer :: Env -> Expr -> Quals -> Maybe Type 
  Infer $\ltenv$ $\ltexpr$ iquals = $\maybeapplysub{\lttyp}{\ltsol}$
    where 
      $\ltsol$ = Solve C $\ltsol_0$
      ($\lttyp$, C) = Cons $\ltenv$ $\ltexpr$ 
\end{mcode}
\begin{mcode}
  Cons  :: Env -> Expr -> (Maybe Type, [Cons])
  Solve :: [Cons] -> Sol -> Maybe Sol  
\end{mcode}
\vspace{-2mm}
\caption{\textbf{Liquid Inference Algorithm} (\texttt{Cons} and \texttt{Solve} are defined in \citep{appendix}).}
\vspace{-4mm}
\label{fig:infer}
\end{figure}
The function \cons{\ltenv}{\ltexpr}
uses the typing rules in Figure~\ref{fig:rules}
to generate the template type @Just@ $\lttyp$ of the expression $\ltexpr$, 
\ie a type that potentially contains liquid variables, 
and the basic constraints that appear in the leaves of the derivation tree
of the judgment \hastype{}{\ltenv}{\ltexpr}{\lttyp}.
If the derivation rules fail, then $\cons{\ltenv}{\ltexpr}$ 
returns @Nothing@ and an empty constraint list. 
The function \solve{C}{\ltsol} uses the decidable validity checking to iteratively 
pick a constraint in $c\in C$ that is not satisfied, 
while such a constraint exists, 
and weakens the solution $\ltsol$ so that $c$ is satisfied. 
The function $\maybeapplysub{\lttyp}{\ltsol}$
applies the solution \ltsol to the type \lttyp, if both 
contain @Just@ values, otherwise returns @Nothing@.
Here and in the following, we pose: 
@ ltsolzero = $\lambda \kvar$.iquals@.

The algorithm \infer{\ltenv}{\ltexpr}{\iquals} is terminating 
and sound and complete with respect to the typing relation 
\hastype{}{\ltenv}{\ltexpr}{\lttyp} as long as all the predicates 
are conjunctions of predicates drawn from \iquals.

\subsection{Gradual Refinement Types}\label{subsec:gradualtypes}
Gradual refinement types~\citep{Lehmann17} extend the refinements of \reflang 
to include imprecise refinements like $x > 0 \wedge \egrad$. While they describe the static and dynamic semantics of gradual refinements, inference is left as an open challenge. 
Our work extends liquid inference to gradual refinements, therefore we hereby summarize their basics.

\paragraph{Syntax}
The syntax of gradual predicates in \graduallang is 
$$
\begin{array}{rcll}
  \gradual{\pred}
  & ::=    & \pred              &\quad \text{Precise Predicate} \\
  &\spmid & \pred \land \egrad &\quad \text{Imprecise Predicates, where } \pred \text{ is local}
\end{array}
$$
A predicate is either {\em precise} or {\em imprecise}. The syntax of an imprecise predicate $\pred \land \egrad$ allows for a {\em static part} $p$. Intuitively, with the predicate $x > 0 \wedge \egrad$, $x$ is statically (and definitely) positive, but the type system can optimistically assume stronger, non-contradictory requirements about $x$. To make this intuition precise and derive the complete static and dynamic semantics of gradual refinements, \citet{Lehmann17} follow the Abstracting Gradual Typing methodology (AGT)~\citep{adt}. Following AGT, a gradual refinement type (resp. predicate) is given meaning by {\em concretization} to the set of static types (resp. predicates) it represents. Defining this concretization requires introducing two important notions.

\paragraph{Specificity}
First, we say that $\pred_1$ is {\em more specific} than $\pred_2$, 
written $\specific{\pred_1}{\pred_2}$,
\textit{iff} $\pred_2$ is true when $\pred_1$ is:
$$
\specific{\pred_1}{\pred_2}
\defeq \forall \esubst . 
\ \text{if}\ \totrue{\applysub{\pred_1}{\esubst}} 
     \ \text{then}\ \totrue{\applysub{\pred_2}{\esubst}}
$$

\paragraph{Locality}
Additionally, in order to prevent imprecise formulas from introducing contradictions---which would defeat the purpose of refinement checking---\citet{Lehmann17} identify the need for the static part of an imprecise refinement to be {\em local}. 
Using an explicit syntax \bpred{\pred}{\rbind} 
to explicitly declare the 
variable $\rbind$ refined by the predicate \pred, a refinement is local 
if there exists a value $\vparam$ for which $\pred\sub{\rbind}{\vparam}$ is true; and this, for any (well-typed) substitution that closes the predicate: 
$$
\islocal{\rbind}{\pred}
\defeq 
\forall \esubst, \exists \vparam. \totrue{\applysub{\pred\sub{\rbind}{\vparam}}{\esubst}}
$$

\paragraph{Concretization}
Using specificity and locality, the concretization function \concrete{\cdot} maps 
gradual predicates to the set of the static predicates they represent. 
\[
\begin{array}{rcl}
\concrete{\bpred{\pred}{\rbind}} &\defeq&  \{ \pred \}  \\
\concrete{\bpred{(\pred \land \egrad)}{\rbind}} &\defeq&
  \{ \pred' \mid \pred' \preceq \pred, \islocal{\rbind}{\pred'} 
  \} 
\end{array}
\]
A precise predicate concretizes to itself (singleton), while an imprecise predicate denotes all the local predicates more specific than its static part.
This definition extends naturally to types and environments.
\[
\begin{array}{rcl}
\concrete{\tref{\rbind}{\btyp}{\gradual{\pred}}} &\defeq&  
  \{ \tref{\rbind}{\btyp}{\pred} \mid \pred\in\concrete{\bpred{\gradual{\pred}}{\rbind}} \}  \\
\concrete{\tfun{x}{\gradual{\tparam_x}}{\gradual{\tparam}}} &\defeq&
  \{ \tfun{x}{\tparam_x}{\tparam} \mid \tparam_x \in \concrete{\gradual{\tparam_x}}, 
                                       \tparam   \in \concrete{\gradual{\tparam}} \} \\
\concrete{\gradual{\envparam}}
&\defeq& \{\envparam \mid 
\bind{x}{\tparam}\in\envparam\ \textit{iff}\
\bind{x}{\gradual{\tparam}}\in\gradual{\envparam},
\gradualinstance{\tparam} \}
\end{array}
\]

The denotations of gradual refinement types are similar to those from \S~\ref{subsec:refinementtypes}. The denotation of a base {\em imprecise} gradual refinement $\tref{\rbind}{\btyp}{\pred \land \egrad}$ includes all (gradually-typed) expressions that satisfy at least $\pred$.

\paragraph{Type Checking}
Figure~\ref{fig:rules} is used ``as is'' to 
type gradual expressions $\hastype{}{\genv}{\gexpr}{\gtyp}$, save for the fact that the validity predicate must be lifted to operate on gradual types. \gradualvalid{\cdot} holds if there exists a justification, by concretization, that the static judgment holds. Precisely:
$$\arraycolsep=0.5pt 
\begin{array}{rcl}
\gradualvalid{\issubtype{}{\genv}{\gtyp_1}{\gtyp_2}}
&\defeq&
\exists
\gradualinstance{\env}, 
\gradualinstance{\typ_1},
\gradualinstance{\typ_2}. 
\valid{\issubtype{}{\env}{\typ_1}{\typ_2}}\\
\gradualvalid{\iswellformed{}{\genv}{\gtyp}}
&\defeq&
\exists
\gradualinstance{\env}, 
\gradualinstance{\typ}.
\valid{\iswellformed{}{\env}{\typ}}
\end{array}
$$

\section{Gradual Liquid Types}
\label{sec:gradual-liquid-types}
\label{sec:theory}

We now formalize the combination of liquid type inference and gradual refinements
to later use gradual liquid types for both error explanation and program migration.
We extend the work of \citet{Lehmann17} by adapting the liquid type inference algorithm to the gradual setting. To do so, we apply the abstract interpretation approach of AGT~\citep{adt} to lift the @Infer@ function (defined in~\S~\ref{subsec:liquidtypes}) so that it operates on gradual liquid types. 

Below is the syntax of predicates in \gradualliquidlang, a gradual liquid core language whose predicates are 
gradual predicates 
where the static part of an imprecise predicate is a liquid predicate,
with the additional requirement that it is {\em local} (def. in \S~\ref{subsec:gradualtypes}).
$$\arraycolsep=0.5pt 
\begin{array}{lcll}
  \glpred
   & ::=    & \lpred              &\quad \text{Precise Liquid Predicate} \\
  &\spmid & \lpred \land \egrad &\quad \text{Imprecise Liquid Predicate, where } \lpred \text{ is local}
\end{array}
$$
The elements of \gradualliquidlang
are both gradual and liquid; \ie~expressions \gliquid{\expr}
could also be written as $\gradual{\liquid{\expr}}$. 
Also, we write $\egrad$ as a shortcut for the imprecise predicate 
@true $\wedge~\egrad$@.

Our goal is to define \ginfer{\glenv}{\glexpr}{\iquals} so that it returns a type $\gltyp$ such that $\hastype{}{\glenv}{\glexpr}{\gltyp}$. After deriving  
\ginfername using AGT (\S~\ref{subsec:adt}), 
we provide an algorithmic characterization of \ginfername (\S~\ref{subsec:algorithmic}), which serves as the basis for our implementation. We present the properties that \ginfername satisfies in \S~\ref{subsec:metatheory}. 

\subsection{Lifting Liquid Inference}\label{subsec:adt}
We define the function $\ginfername$ using the abstracting gradual typing 
methodology~\citep{adt}.
In general, AGT defines the consistent lifting of a function {\tt f} as:
$\ttgradual{f}\ \gradual{t} = \alpha(\{\texttt{f}\ t \spmid t \in\concrete{\gradual{t}}\})$, 
where $\alpha$ is the sound and optimal abstraction function that, together with $\gamma$, forms a Galois connection.

The question is how to apply this general approach to the liquid type inference algorithm. We answer this question via trial-and-error.

\paragraph{Try 1. Lifting {\tt Infer}} 
Assume we lift @Infer@ in a similar manner, \ie~we pose
$$
\ginfer{\glenv}{\glexpr}{\iquals} = \alpha(\{\infer{\env}{\expr}{\iquals} \spmid 
\lgradualinstance{\env}, 
\lgradualinstance{\expr}\})
$$
This definition of $\ginfername$ is too strict: 
it rejects expressions that should be accepted. 
Consider for instance the following expression $\glexpr$
that defines a function @f@ with an imprecisely-refined argument:
\begin{code}
  // onlyPos :: {v:Int | 0 < v} -> Int
  // check :: Int -> Bool
  let f :: x:{Int | ?} -> Int 
      f x = if check x then onlyPos x else onlyPos (-x)
  in f 42
\end{code}
There is no single static expression $\lgradualinstance{\expr}$
such that the definition of @f@ above type checks. 
For any $\iquals$ we will get 
$\infer{\{\}}{\expr}{\iquals} = \texttt{Nothing}$ which denotes a type inference failure.  
This behavior occurs because by the definition of $\ginfername$
the gradual argument of @f@ needs  to be concretized before calling $\infername$,
which breaks the flexibility programmers expect from gradual refinements (this example is based on the motivation example of \citet{Lehmann17}). One expects that 
$\ginfer{\{\}}{\glexpr}{\iquals}$ should simply return @Int@.
A similar scenario appears in \citet{adt}, where lifting the typing relation as a whole
would be too imprecise and instead the 
lifted typing relation is defined by lifting the type functions and predicates 
used to define typing.
Here, as described in \S~\ref{subsec:liquidtypes}, @Infer@ calls the functions @Cons@ and @Solve@, 
which in turn calls the function @isValid@. Which of these functions should we lift?
Since @Cons@ is merely an algorithmic definition of the typing rules, 
it is not affected by the gradualization of the system, thus does not require lifting.
On the contrary, @Solve@ is calling @isValid@ that operates on gradual refinements. 
To get a precise inference system we first attempted to lift @isValid@.

\paragraph{Try 2. Lifting {\tt isValid}}
To our surprise, using gradual validity checking 
(the lifting of @isValid@, \S~\ref{subsec:gradualtypes}) leads to an unsound inference algorithm!
This is because, soundness of static inference
implicitly relies on the property of validity checking
that if two refinements $\pred_1$ and $\pred_2$
are right-hand-side valid, then so is their conjunction, \ie:
$$
\begin{array}{c}
\valid{\issubtype{}{\env}{\tref{v}{\btyp}{\pred}}{\tref{v}{\btyp}{\pred_1}}}\qquad
\text{and}\qquad\valid{\issubtype{}{\env}{\tref{v}{\btyp}{\pred}}{\tref{v}{\btyp}{\pred_2}}}\\
\Rightarrow\\
\valid{\issubtype{}{\env}{\tref{v}{\btyp}{\pred}}{\pred_1 \land \pred_2}}
\end{array}
$$
But this property does not hold for gradual validity checking, 
because for any logical predicate $q$, it is true that 
$(q \Rightarrow p_1 \land q \Rightarrow p_1)$
implies $(q \Rightarrow p_1 \land p_2) $, but 
$(\exists q. (q \Rightarrow p_1)) \land (\exists q. (q \Rightarrow p_1))$ does not imply that $\exists q. (q \Rightarrow p_1 \land p_2)$.

\paragraph{Try 3. Lifting {\tt Solve}}
Let us try to lift @Solve@: 
\begin{mcode}
  $\gsolve{\ltsol}{\gliquid{C}} = \{\solve{\ltsol}{C} \spmid \lgradualinstance{C}\}$
\end{mcode}
where $\gliquid{C}$ denotes a gradual constraint (from Figure~\ref{fig:syntax}).
This approach is successful and leads to a provably sound and complete inference algorithm (\S~\ref{subsec:metatheory}). 

Note that in the definition of \gsolvename, we do not appeal to abstraction. This is because we can directly define \ginfername to consider all produced solutions instead.
\begin{mcode}
  $\ginfername$ :: $\glenvtyp$ -> $\glexprtyp$ -> Quals -> Set $\texttt{Type}$ 
  $\ginfer{\glenv}{\glexpr}{\iquals}$ = {$\gltyp'$ | Just $\gltyp'$ <- $\maybeapplysub{\gltyp}{\ltsol}$, $\ltsol \in \gsolve{\ltsol_0}{\gliquid{C}}$}
    where ($\gltyp$, $\gliquid{C}$) = Cons $\glenv$ $\glexpr$ 
\end{mcode}
First, function @Cons@
derives the typing constraints $\gliquid{C}$, and if successful, 
the template type $\gltyp$ (step 1 of~\S~\ref{subsec:overview:glt}).\footnote{Note that \texttt{Cons} is unchanged from the static system, because it only depends on the structure of the types, and not on the refinements themselves.}
Next, we use the lifted \gsolvename 
to concretize and solve all the derived constraints 
(steps 2 and 3 of~\S~\ref{subsec:overview:glt}, \resp). 
By keeping track of the concretizations 
that return non-@Nothing@ solutions, we derive the
safe concretizations of~\S~\ref{sec:overview}. 
Finally, we apply each solution
to the template gradual type, yielding a set of inferred types.
We do not explicitly abstract the set of inferred types back to a single gradual type. Adding abstraction by exploiting the abstraction function defined by \citet{Lehmann17} is left for future work. 

The core of the gradual inference algorithm is a combination 
of the liquid and gradual refinement type systems. 
Yet, as we exposed by the failing attempt to lift @isValid@, 
this combination was not trivial, since 
blind application of the AGT~\citep{adt} methodology could lead to unsound inference. 
In~\S~\ref{subsec:metatheory}, we prove that our algorithm is sound 
and in~\S~\ref{sec:implementation} we discuss an optimized implementation 
and its applications.
The applications discussed in \S~\ref{sec:error-explanation} and \S~\ref{sec:migration} make explicit use of the inferred set in order to assist users in understanding errors and migrating programs.\footnote{In standard gradual typing, the set of static types denoted by a gradual type can be infinite, hence abstraction is definitely required. In contrast, here the structure of types is fixed, and the set of possible liquid refinements, even if potentially large, is finite. We can therefore do without abstraction. We discuss implementation considerations in \S~\ref{sec:error-explanation}.}

\subsection{When Does Gradual Liquid Inference Fail?}

Given the flexibility induced by gradual refinements, at this point the reader might wonder when an actual type inference failure can occur.
First, as we will formally prove later in this section, gradual liquid inference is a conservative extension of liquid inference, and therefore, in the absence of imprecise refinements, gradual liquid inference fails exactly when standard liquid inference fails.

More crucially, in presence of imprecise refinements, gradual inference only succeeds when there exist possible justifications for each individual occurrences of gradually-refined variables. Said otherwise, gradual liquid inference fails when there exists at least one occurrence of a gradually-refined variable for which there does not exist any valid concretization.

Consider the example below:
\begin{mcode}
  // onlyPos :: {v:Int | 0 < v} -> Int
  f :: x:{Int | x <= 0 $\land$ ?} -> Int 
  f x = onlyPos x 
\end{mcode}
Each valid concretization of the imprecise refinement
@x <= 0 $\land$ ?@ should imply the precondition of @onlyPos@, namely 
@0 < v@. This however contradicts the static part of the imprecise refinement, @x <= 0@. 
Therefore, gradual liquid inference fails for this program.

\subsection{Algorithmic Concretization}\label{subsec:algorithmic}
To make \ginfername algorithmic, we need to define an algorithmic concretization function 
of a set of constraints, $\concrete{\gliquid{C}}$. 
We do so, by using the the finite domain of predicates $\iquals$.

Recall the concretization of gradual 
predicates: %
$
\concrete{\bpred{(\pred \land \egrad)}{\rbind}}$ $\defeq
  \{ \pred' \mid \pred' \preceq \pred, \islocal{\rbind}{\pred'} 
  \}  
$.
In general, this function cannot be algorithmically 
computed, since it ranges over the infinite domain of predicates. 
In gradual liquid refinements, 
the domain of predicates is restricted to the powerset of the 
{\em finite} domain $\iquals$. 
We define the algorithmic concretization function
$\lconcrete{\bpred{\glpred}{\rbind}}$ as the intersection of 
the powerset of the finite domain $\iquals$ with the gradual concretization function.  
$$
\lconcrete{\bpred{\glpred}{\rbind}} \defeq
  \powerset{\iquals} \cap \concrete{\bpred{\glpred}{\rbind}}
$$
Concretization of gradual predicates 
reduces to (decidable) locality and specificity checking on 
the elements of $\iquals$.
$$
\lconcrete{\bpred{(\lpred \land \egrad)}{\rbind}} \defeq
  \{ \lpred' \mid \lpred'\in\powerset{\iquals}, \lpred' \preceq \lpred, \islocal{\rbind}{\lpred'} 
  \}  
$$
We naturally extend the algorithmic concretization function 
to typing environments, constraints, and list of constraints.
\[ \arraycolsep=1pt 
\begin{array}{rcl}
\lconcrete{\gliquid{\env}}
&\defeq& \{\lenv \mid 
\bind{x}{\ltyp}\in\lenv \ \textit{iff}\
\bind{x}{\gliquid{\typ}}\in\gliquid{\env},
\algradualinstance{\typ} \}\\
\lconcrete{\issubtype{}{\glenv}{\gltyp_1}{\gltyp_2}}
&\defeq& \{\issubtype{}{\lenv}{\ltyp_1}{\ltyp_2} \mid 
\algradualinstance{\env},
\algradualinstance{\typ_i} \}\\
\lconcrete{\iswellformed{}{\glenv}{\gltyp}}
&\defeq& \{\iswellformed{}{\lenv}{\ltyp} \mid 
\algradualinstance{\env},
\algradualinstance{\typ}, \}\\
\lconcrete{\gliquid{C}}
&\defeq& \{\liquid{C} \mid 
c\in \liquid{C}\ \textit{iff}\
\gliquid{c}\in\gliquid{C},
\algradualinstance{c} \}
\end{array}
\]
We use \lconcrete{\cdot}
to define an algorithmic version of \gsolvename: 
\begin{mcode}
  $\gsolve{\ltsol}{\gliquid{C}} = \{\solve{\ltsol}{\liquid{C}} \spmid \algradualinstance{C}\}$
\end{mcode}
which in turn yields an algorithmic version of \ginfername.

\subsection{Properties of Gradual Liquid Inference}\label{subsec:metatheory}

We prove that the inference algorithm \ginfername satisfies the  
the correctness criteria of~\citet{LiquidPLDI08}, as well as the static 
criteria for gradually-typed languages~\citep{siek15}.\footnote{Because this work focuses on the static semantics, \ie~the inference algorithm, we do not discuss the dynamic part of the gradual guarantee, which has been proven for gradual refinement types~\citep{Lehmann17}.} The corresponding proofs can be found in supplementary material~\citep{appendix}.

\subsubsection{Correctness of Inference}
The algorithm \ginfername is sound, complete, and terminates.

\begin{theorem}[Correctness]~\label{theorem:correctness}
Let $\iquals$ be a finite set of predicates from an SMT-decidable logic, 
\glenv a gradual liquid environment, and
\glexpr\;a gradual liquid expression. Then 
\begin{itemize}[leftmargin=*]
\item\textbf{Termination}
\ginfer{\glenv}{\glexpr}{\iquals} terminates.
\item\textbf{Soundness} 
If $\gltyp \in \ginfer{\glenv}{\glexpr}{\iquals}$, 
then \hastype{}{\glenv}{\glexpr}{\gltyp}. 
\item\textbf{Completeness} 
If $\ginfer{\glenv}{\glexpr}{\iquals} = \emptyset$, 
then $ \not \exists\gltyp.\
\hastype{}{\glenv}{\glexpr}{\gltyp}$.
\end{itemize}
\end{theorem}
The proof of termination is straightforward from the careful definition of the 
$\ginfername$ algorithm. Soundness and completeness rely on 
the property that a constraint is gradually valid \textit{iff} there 
exists a concrete solution that renders it valid (\ie
\gradualvalid{\applysub{\gliquid{c}}{\ltsol}}
\textit{iff}
$\exists \algradualinstance{c}. \valid{\applysub{\gliquid{c}}{\ltsol}}$). 
We use this property to prove soundness by the definition of the algorithm 
and completeness by contradiction.
We note that, 
unlike the \infername algorithm that provably returns the strongest possible solution, 
it is not clear how to relate the set of solutions 
returned by \ginfer{\glenv}{\glexpr}{\iquals} with
the rest of the types that satisfy \hastype{}{\glenv}{\glexpr}{\gltyp}.

\subsubsection{Gradual Typing Criteria} 

\citet{siek15} list three criteria for the static semantics of a gradual language, which the \ginfername algorithm satisfies. 
These criteria require the gradual type system 
\emph{(i)} is a conservative extension of the static type system, 
\emph{(ii)} is flexible enough to accommodate the dynamic end of the typing spectrum (in our case, unrefined types), and \emph{(iii)} supports a smooth connection between both ends of the spectrum.

\paragraph{(i) Conservative Extension}
The gradual inference algorithm \ginfername
coincides with the static algorithm \infername on terms that only rely on precise predicates.
More specifically, 
if \infername infers a static type for a term, 
then \ginfername returns only that type, for the same term. 
Conversely, if a term is not typeable with \infername, it is also not typeable with \ginfername.

\begin{theorem}[Conservative Extension]
If $\infer{\lenv}{\lexpr}{\iquals} = \texttt{Just}\ \ltyp$, 
then $\ginfer{\lenv}{\lexpr}{\iquals} = \{\ltyp\}$. 
Otherwise, 
$\ginfer{\lenv}{\lexpr}{\iquals} = \emptyset$. 
\end{theorem}

The proof follows by the definition of $\ginfer$ and the concretization function.

\paragraph{(ii) Embedding of imprecise terms}

We then prove that given a well-typed unrefined term (\ie~simply-typed), refining all base types with the unknown predicate\egrad yields a well-typed gradual term. This property captures the fact that it is possible to ``import'' a simply-typed term into the gradual liquid setting. (In contrast, this is not possible without gradual refinements: just putting \texttt{true} refinements to all base types does not yield a well-typed program.)

To state this theorem, we use $t_s$ to denote simple types ($\btyp$ and $t_s \rightarrow t_s$) and similarly $e_s$ and $\Gamma_s$ for terms and environments. The simply-typed judgment is the standard one.

The $\dyn\cdot$ function turns simple types into gradual liquid  types by introducing the unknown predicate on every base type (and naturally extended to environments and terms):
\begin{align*}
\dyn{\btyp} &= \tref{v}{\btyp}{\egrad} &
\dyn{\typ_1 \rightarrow \typ_2} &= \tfun{x}{\dyn{\typ_1}}{\dyn{\typ_2}}
\end{align*}
\begin{theorem}[Embedding of Unrefined Terms]\label{thm:embedding}
If $\hastype{}{\env_s}{\expr_s}{\typ_s}$, 
then $\ginfer{\dyn{\env_s}}{\dyn{\expr_s}}{\iquals}\not = \emptyset$.
\end{theorem}

Let $\typ_s$
denote the unrefined version of the gradually refined type $\gltyp$. 
We first prove that for any gradual type $\gltyp$ with only local refinement 
both \issubtype{}{\env}{\gltyp}{\dyn{\typ_s}} and 
\issubtype{}{\env}{\dyn{\typ_s}}{\gltyp} hold.
Then, we prove that 
if $\hastype{}{\env_s}{\expr_s}{\typ_s}$, 
then $\hastype{}{\dyn{\env_s}}{\dyn{\expr_s}}{\dyn{\typ_s}}$.
Finally, we use completeness of the inference algorithm to prove theorem~\ref{thm:embedding}.

For the expression $\dyn{\expr_s}$, the inference algorithm 
will generate a set of refinement types, all of which have the structure  
of $\typ_s$. 
The liquid algorithm infers types with strongest postconditions
and the gradual inference algorithm is not generating any fresh @?@,  
so, the inference returns all valid types with shape $\typ_s$
with the dynamic refinements instantiated to all valid local predicates
and for each instantiation, the unknown refinements are solved 
following the principle of strongest postconditions.

\paragraph{(iii) Static Gradual Guarantee}
Finally, the solid foundations for gradualization based on abstract interpretation allow us to effectively satisfy the {\em gradual guarantee}, which we believe was never proved for any gradual inference work. The gradual guarantee stipulates that typeability is monotonic in the precision of type information. In other words, making type annotations less precise cannot introduce new type errors. 

We first define the notion of precision in terms of algorithmic concretization:
\begin{definition}[Precision of Gradual Types]
$\gltyp_1$ is less precise than $\gltyp_2$, written as 
$\gless{\gltyp_1}{\gltyp_2}$, \textit{iff}
$\lconcrete{\gltyp_1} \subseteq  \lconcrete{\gltyp_2}$.
\end{definition}
\noindent
Precision naturally extends to type environments and terms.
\begin{theorem}[Static Gradual Guarantee]\label{theorem:gradual-guarantee}
If 
$\glenv_1\sqsubseteq\glenv_2$ and
$\glexpr_1\sqsubseteq\glexpr_2$,
then 
for every 
$\gltyp_{1i}\in\ginfer{\glenv_1}{\glexpr_1}{\iquals}$
there exists \gless{\gltyp_{1i}}{\gltyp_{2i}} so that 
$\gltyp_{2i}\in \ginfer{\glenv_2}{\gltyp_2}{\iquals}$.
\end{theorem}
\noindent

The intuition of the proof is that
since $\gltyp_{1i}\in\ginfer{\glenv_1}{\glexpr_1}{\iquals}$, 
then the algorithm inferred a solution of the unknown refinements. 
We prove that this solution is also inferred for $\glexpr_1$.

For a given term and type environment, 
the theorem ensures that, for every inferred type, 
the algorithm infers a less precise type when run on a less precise term and environment.

\section{Implementation}
\label{sec:implementation}

We implemented \ginfername as \toolname, 
an extension to \liquidHaskell~\citep{Vazou2014LiquidHaskellEW}
that takes a Haskell program annotated with gradual 
refinement specifications and returns an @.html@ 
interactive file that lets the user explore all safe concretizations.

Concretely, \toolname 
uses the existing API of \liquidHaskell to 
implement the three steps of gradual liquid type checking 
steps described in~\S~\ref{subsec:overview:glt} 
(and formalized in~\S~\ref{sec:theory}):
\textit{1)} First, \toolname calls the \liquidHaskell API 
to generate subtyping constraints 
that contain both liquid variables and imprecise predicates.
\textit{2)} Next, it calls the liquid API to collect all the refinement templates. The templates are used to map each occurrence of imprecise predicates in the constraints to 
a set of concretizations. These concretizations are combined 
to generate the possible concretizations of the constraints. 
\textit{3)} Finally, using \liquidHaskell's constraint solving it decides 
the validity of each concretized constraint, while all the safe concretizations \scshorts are interactively presented to the user.  

The implementation of \toolname closely follows the theory of \S~\ref{sec:theory}, 
apart from a syntactic detail (we use @??@ instead of @?@ to denote the unknown part of a
refinement), 
and several practical adjustments that we discuss here and 
are crucial for real-world applicability.

\paragraph{Templates}
To generate the refinement templates we use 
\liquidHaskell's existing API.
The generated templates consist of 
a predefined set of predicates for 
linear arithmetic ($v\ [< | \leq | > | \geq | =  | \not =]\ x$), 
comparison with zero ($ v\ [< | \leq | > | \geq | =  | \not =]\ 0$) 
and length operations ($\text{len}\ v [\geq | > ] 0$, $\text{len}\ v = \text{len}\ x$, $v = \text{len}\ x$, $v = \text{len}\ x + 1$), 
where @v@ and @x@ respectively range over the refinement variable 
and any program variable.
Application-specific templates are automatically abstracted from user-provided specifications
and the user can explicitly define custom templates. 
For instance, 
assume a user defined data type @Tree a@
and an uninterpreted function @size :: Tree a -> Int@.
If the user writes a specification 
@x: Tree a -> {v:Int | v < size x}@, 
then the template @v < size x@ is generated 
where @v@ and @x@ respectively range over the refinement variable 
and any program variable.

\paragraph{Depth}

For completeness in the theory,
we check the validity of any solution, including all possible {\em conjunctions} of elements of $\quals$, which is not tractable in practice. 
The implementation uses an instantiation depth parameter
that is 1 by default, meaning each @??@ 
ranges over single templates.
At depth 2, @??@ ranges over conjunctions of (single) templates, \textit{etc.}

\paragraph{Sensibility Checking}
Each @??@ can be instantiated 
with any templates that are local and specific.
The implementation uses the SMT solver to check both. 
As an \textit{optimization}, we perform a 
syntactic locality check to reject templates that are ``non-sensible'', 
for instance, syntactic contradictions of the form @x < v && v < x@.
As a \textit{heuristic}, we further filter out as non-sensible 
type-directed instantiations of the templates that, 
based on our experience, 
a user would not write, such as arithmetic operations on lists and booleans (\eg @x < False@)---although potentially correct in Haskell through overloading.

\paragraph{Locality Checking}
To encode Haskell functions (\eg @len@) in the refinements, 
\liquidHaskell is using uninterpreted SMT functions~\citep{Vazou18}.
As~\citet{Lehmann17} note, locality checking breaks under the presence of uninterpreted functions. 
For instance, the predicate @0 < len i@ is not local on @i@,
because @$\exists$i. 0 < len i@ is not SMT valid due to a model 
in which @len@ is always negative. 
To check locality under uninterpreted functions
we define a fresh variable (\eg @leni@)
for each function application (\eg @len i@). 
For instance, @0 < len i@ is local on @i@ because 
under the new encoding 
@$\exists$leni. 0 < leni@ is SMT valid.

\paragraph{Partitions}
A critical optimization for efficiency is that
after generation and before solving, 
the set of constraints is partitioned 
based on the constraint dependencies 
so that each partition is solved independently.
Two constraints depend on each other when they contain the same liquid variable
or when they contain different variables (\eg @k$_1$@ and @k$_2$@)
that depend on each other 
(\eg @x:{k$_1$} $\vdash${v|true}$\preceq${v|k$_2$}@).
That way, we reduce the number of @??@ that appear in each set of independent constraints %
and thus the number of concretization combinations that need to be checked 
(which increases exponentially with the number of @??@ accounting for all combinations of all concretizations).

\section{Application I: Error Explanation}\label{sec:error-explanation}

\begin{figure*}
  \centering
    \includegraphics[width=0.45\textwidth]{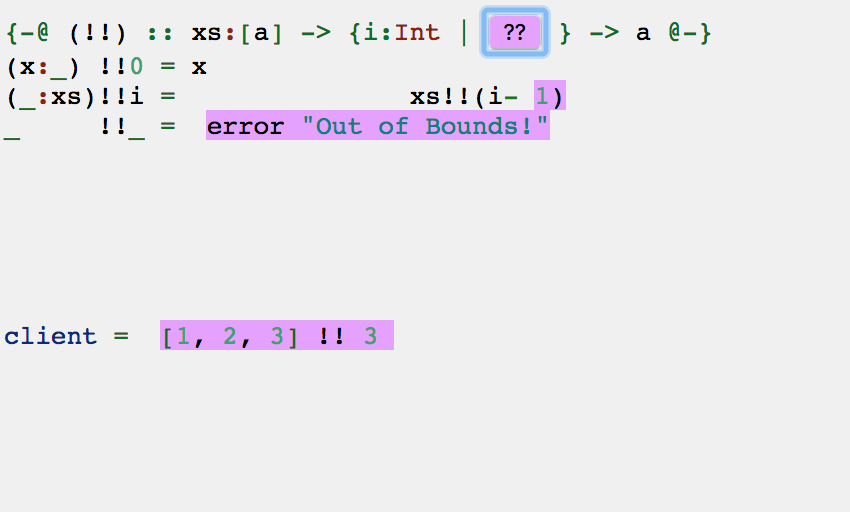}
    \hfill
    \includegraphics[width=0.45\textwidth]{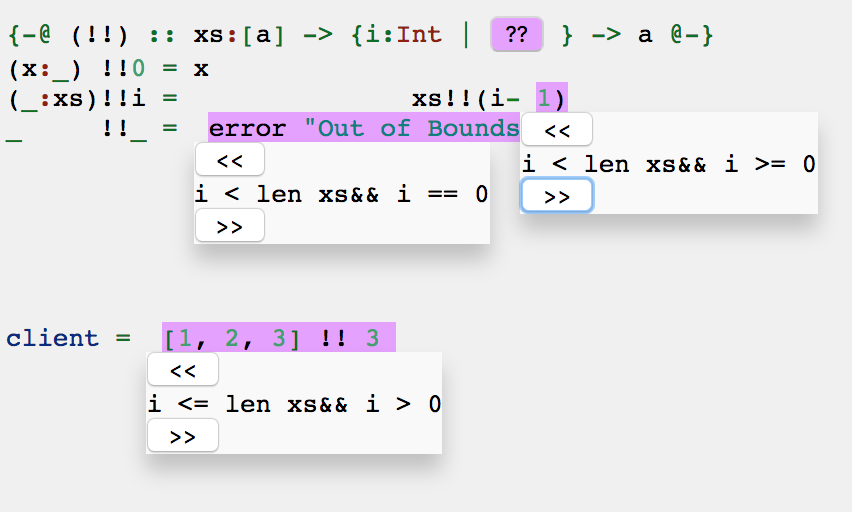}
\caption{\toolname GUI for error explanation of incompatible indexing.}
\label{figure:gui}
\end{figure*}

We illustrate how \toolname can be used interactively for error explanation.
Consider the list indexing function @(!!)@ in Haskell:
\begin{mcode}
  (x:_)  !! 0 = x 
  (_:xs) !! i = xs!!(i - 1)
  _      !! _ = error "Out of Bounds!"
\end{mcode}
Indexing is @0@-based and signals
a runtime error for @length xs <= i@.
Now consider a client indexing a list by its length:
\begin{code}
  client = [1, 2, 3] !! 3 
\end{code}
Refinement types can be used to impose a false precondition on the @error@ 
function and detect the out-of-bound crash:
\begin{mcode}
  error :: {i:String | false} -> a
\end{mcode}
But when the two above incompatible definitions coexist is the error
due to the {\em definition} of indexing, or to the {\em client} of indexing? 

This question of who is to blame has no definitive answer in the general case. 
Gradual liquid type inference is helpful to explore possible resolutions of
the error and decide for a suitable solution. 
To do so, we transform the statically ill-typed program
into a gradually well-typed program by giving an imprecise refinement 
as a precondition to the indexing function:
\begin{mcode}
  (!!) :: xs:[a] -> {i:Int | ?? } -> a
\end{mcode}
Using \toolname, we can explore all potential predicates 
that can be substituted for @??@. 
Among the candidates are both @0 <= i < len xs@ and @0 < i <= len xs@. While the former induces a type error in the client, the latter induces a type error in the definition of @(!!)@.

Figure~\ref{figure:gui}
was generated after running \toolname on 
the above specification and code.
The result of \toolname is an @.html@ file 
where each user specified @??@ is turned into a colored button. 
Pressing a @??@ button once 
highlights all of its usage occurrence, using different colors for each @??@ (Figure~\ref{figure:gui}(left)).
Pressing it again 
presents all safe concretizations to scroll through (Figure~\ref{figure:gui}(right)).
In the indexing example, three gradual constraints are generated: 
\textit{1)} for the recursive call of indexing, 
\textit{2)} for the unreachable (due to the false precondition) error call, and 
\textit{3)} for the client.
Unsurprisingly, 
the \scshorts for the recursive call and the client 
include the incompatible @0 <= i < len xs@ and @0 < i <= len xs@, respectively.
Interestingly, the unreachable constraint enjoys many \scshorts
including @0 <= i <= len xs@ and 
one given in the right of Figure~\ref{figure:gui}, 
\ie @i < len xs && i == 0@,
since the case of indexing @0@ from a non-empty list is covered 
in the first case of the indexing function.

The user explores all \scshorts with the @<<@ and @>>@ buttons.
In the example, the three \scshorts are independent, but in many cases
\scshorts can depend on each other (\eg dependencies in function
preconditions).
In such cases, pressing a navigation button changes the values of all the dependent occurrences.

\begin{figure*}
  \centering
    \includegraphics[width=0.45\textwidth]{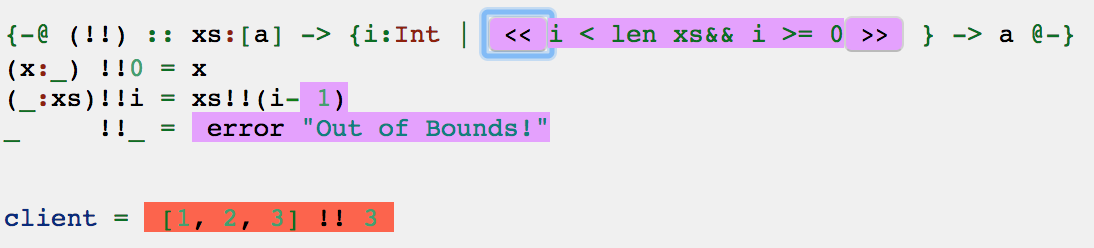} 
    \hfill
    \includegraphics[width=0.45\textwidth]{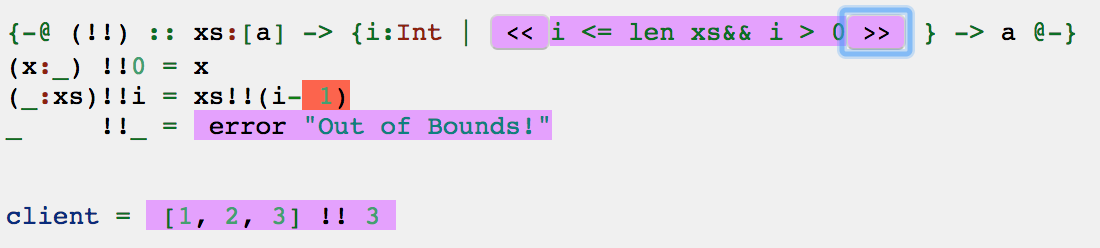}
  \caption{\toolname GUI for liquid types exploration.}
  \label{figure:guialt}
\end{figure*}

When the goal is to replace the @??@ with a concrete refinement, 
going through all \scshorts can be overwhelming. 
In the example, there are 
22, 10, and 13 \scshorts
for the 
unreachable, client, and recursive calls, respectively.
To accommodate the user, \toolname generates an alternative 
interactive @.out.html@ file by which the user can 
replace the @??@ with any \scshort and observe 
the generated refinement errors. 
In Figure~\ref{figure:guialt}-left
the @??@ is replaced with the concrete refinement @0 <= i < len xs@,
generating an error at the client site. 
Pressing @>>@, the user explores the next concrete refinement, where 
@0 < i <= len xs@ generates an error in the recursive case of indexing (Figure~\ref{figure:guialt}-right). 
This exposes all the concretizations of @??@ 
that render at least one constraint safe (here 22).

\paragraph{Gradual types for Error Explanation}
This interactive replacement of @??@ allows the user to try 
all the possible concretizations inferred by the algorithm. 
For each selection the associated type errors are generated and highlighted. 
This process is effective for two major reasons:
\begin{enumerate}[leftmargin=*]
\item \textit{Provide all error locations:} 
The inference algorithm does not arbitrarily blame one program location
based on one inferred predicate, 
instead the user navigates via \textit{all} the potential error locations. 
Thus, during navigation the user can observe all potential sources of program errors.
\item \textit{Predicate Synthesis:}
Moreover, the user does not have to come up with the refinements, 
since the inference procedure synthesizes them. 
It is only up to the user to pick the adequate refinement that she chooses 
as the desired specification---which could be fully static, or imprecise.\footnote{For sound imprecise refinements, runtime support for gradual refinements is needed, as described by \citet{Lehmann17}. Implementing such support in Liquid Haskell is an interesting perspective, outside the scope of this work.}
\end{enumerate}

\paragraph{Quantitative Evaluation}
\begin{table*}
\begin{center}
\begin{tabular}{ | c | c || c | c | c | c | c || c | c | c | c || r |}
\hline
\depth & \gradNum  & \occs & \cands & \sens & \local & \precise & 
\parts & \instan & \sols & \staticsols & \ttime (s)\\          
    \hline
	1 &	1 &	[3] &	[12*] &	[11*] &	[11*] &	[11*] &	3/12 &	11* &	[8,0,6] &	0
	& 0.47\\
	2 &	1 &	[3] &	[68*] &	[38*] &	[34*] &	[34*] &	3/12 &	34* &	[22,10,13] &	0
    & 4.91\\
\hline
\end{tabular}
\end{center}
\caption{Quantitative Evaluation of the Indexing Example}
\vspace{-4mm}
\label{table:indexing}
\end{table*}
Table~\ref{table:indexing}
summarizes a quantitative evaluation of the indexing example. 
We run \toolname with instantiation depth (\depth) 1 and 2, 
\ie the size of the conjunctions of the templates we consider (\S~\ref{sec:implementation}).
The \gradNum column gives the number of imprecise refinements 
(as added by the user)
and \occs gives the number each @??@
appearing in the generated constraints. 
Column \cands denotes the candidate (\ie well-typed) 
templates for each occurrence (\S~\ref{sec:implementation}). 
In the example, for depth 2, 68 templates are generated for each occurrence 
(\ie [68, 68, 68] simplified as 68* for space).
Then, \toolname decides how many of the templates are sensible
(\sens), local (\local), and specific (\precise); here, 38, 34, and
34, respectively.
We note that all the local templates are specific, 
when the static part of the gradual refinement is true 
(\ie @true && ??@, simplified as @??@). 
Moreover, note that most non-local solutions were filtered out 
by the sensibility check.
The constraints were split in 12 partitions (\parts), 
out of which 3 contained gradual refinements and had to be concretized. 
Each partition had 34 concretizations (\instan)
out of which 22, 10, and 13 were safe concretizations (\sols), respectively.
None of these concretizations was common for all the occurrences of the @??@, 
thus the \toolname reports 0 static solutions (\staticsols), as expected. 
Finally, the running time of \toolname was 4.91 sec. 

In short, we observe how the implementation optimizations, 
as discussed in~\S~\ref{sec:implementation},
allow for possible error reporting using gradual liquid types. 
With depth of 2 and 3 occurrences of the @?@ each with 68 candidates, 
the exponential inference algorithm would need to run $68^3$ times. 
In practice, we filter the candidates to 34 and since each occurrence
appears in a different partition, we only run the algorithm $34 * 3$
times.   
Next, we use these metrics to 
evaluate \toolname as a possible tool for migrating 
three standard Haskell list libraries to \liquidHaskell.

\section{Application II: Migration Assistance}\label{sec:migration}
As a second application, we use
\toolname's error reporting GUI from \S~\ref{sec:error-explanation} to assist the
migration of commonly used Haskell libraries to Liquid
Haskell.  
Inter-language migration to strengthen type safety
guarantees is one of the main motivations for gradual type systems~\citep{STH-DLS06-scripts-to-programs}.
Our study confirms that gradual liquid type inference provides an effective
bridge to migrate programs written in a traditional polymorphic functional language like Haskell to a stronger discipline with refinement types as provided by Liquid Haskell.

\paragraph{Benchmarks}
We used  \toolname to migrate three
interdependent Haskell list 
libraries: %
\begin{itemize}%
\item 
\bench{GHC.List}: %
provides commonly used list functions available at Haskell's \bench{Prelude},
\item 
\bench{Data.List}: %
defines more sophisticated list functions, \eg list transposing, and
\item 
\bench{Data.List.NonEmpty}: %
lifts list functions to a non-empty list data type. 
\end{itemize}

\paragraph{Migration Process}
A library consists of a set of function 
imports and definitions.  
Each function
comes with its Haskell type
and may be assigned to a gradual refinement type
during migration. 
Migration is complete when, if possible,   
all functions are given fully-static refinement types,
and type check under \liquidHaskell.   
The process proceeds in four steps:

\begin{description}%
\item [Step 1:] run \liquidHaskell to generate a set of
type errors,
\item [Step 2:] fix the errors by \textit{manually} inserting gradual refinements (@??@),
\item [Step 3:] use \toolname to replace @??@ with an \textit{automatically} generated \scshort, and 
\item [Step 4:] go back to \textbf{Step 1} until no type errors are reported.
\end{description}
This process is iterative and interactive since refinement errors
propagate between imported and client libraries 
and it is
up to the user to decide how to resolve these errors.

\textbf{Step 1:}
At the beginning of the migration process 
the source files given to \liquidHaskell 
have no refinement type specifications. 
Still, there are 
two sources\footnote{Termination checking, by default activate in \liquidHaskell,
is a third source of type errors in programs without any refinement type specifications. 
To simplify our case study we deactivated termination checking.}
of %
refinement type errors:
\textit{1)} failure to satisfy imported functions preconditions, 
\eg in \S~\ref{sec:error-explanation} 
the \bench{error} function assumes the \bench{false} precondition and 
\textit{2)} incomplete patterns, \ie a pattern-match that might fail at runtime, 
\eg \bench{scanr}'s result is matched to a non-empty list.
\begin{mcode}
  scanr _ q []     =  [q]
  scanr f q (x:xs) =  f x q : qs where qs@(q:_) = scanr f q xs 
\end{mcode}
Each time \textbf{Step 1} is reiterated, new type errors can occur as a consequence of new function preconditions added in the following steps. 

\textbf{Step 2:}
The insertion of gradual refinements (@??@) 
is left to the user; they can add @??@
either in the preconditions of defined functions 
or postconditions of imported functions,
thereby resolving the type errors of \textbf{Step 1}.
If the user places redundant @??@, then they will be solved to the @True@ static refinement.
If the user misses some @??@, type errors will remain.

\textbf{Step 3:}
Once the program gradually typechecks, 
the user explores the generated
\scshorts, and chooses which one to replace @??@ with. 
In our benchmarks, often, the decision is trivial
since only one \scshort coincides for all occurrences of the gradually-refined variable.

\textbf{Step 4:}
Depending on the refinement inserted at \textbf{Step 3}, new errors might appear in the current or imported libraries; there are three possible scenarios:
\begin{enumerate}[label=(\alph*)]
  \item If the user refined the precondition of a function (\eg \bench{head} requires non-empty lists), then a type error might be generated 
  at clients of the function both inside and outside the
  library;
  \item If the user refined the post-condition of a function (\eg \bench{scanr} always returns non-empty lists), then no error can be generated;
  \item If the user refined the postcondition of an imported function upon which the function to be verified relies, 
  then the imported function's specification may not be satisfied by its implementation. In this case the user can either {\em assume} the
  imported type (thus relying on gradual checking), or update and re-check the imported library.
\end{enumerate}

\begin{small}
\begin{table*}[!h]
\begin{center}
\begin{small}
  \begin{tabular}{ | c || l ||@{ }c@{ }| c | c | c || c | c | c |@{ }c@{ }|| r |}
    \hline
\round & \funname \deps{& Dependencies}
\fdepth{& \depth} & \gradNum  & \occs & \cands &  \precise & 
\parts & \instan & \sols & \staticsols & \ttime (s)\\          
    \hline\hline
\multicolumn{11}{|l|}{\ghclist (56 functions defined and verified)}\\
	\hline
$1^{\text{st}}$      
	& \errorEmptyList & 
	\deps{- &} 
    \fdepth{1 &}	1 &	[1] &	[[5]] &		[[4]] &	1/4 &	[4] &	[0] & 0
    & 1.00
    \\
   	& \bench{scanr}   	&
   	\deps{- &}
	\fdepth{1 &}	1 &	[4] &	[6*] &		[5*] &	1/5 &	[625] &	[125] & 1
	& \textcolor{red}{4.6K}
    \\
   	& \bench{scanr1}   	& 
   	\deps{\errorEmptyList &}
    \fdepth{1 &}	2 &	[4,6] &	[12*,5*] &		[5*,4*] &	1/5 &	[2.5M] & \todo & \todo
    & \textcolor{red}{N/A}
    \\
	\hline
$2^{\text{nd}}$ 
	& \bench{head} 		& 
	\deps{\errorEmptyList &}
	\fdepth{1 &}	1 &	[1] &	[[5]] &	[[4]] &	1/3 &	[4] &	[1]& 1
	& 0.70   
	\\
   	& \bench{tail} 		& 
   	\deps{\errorEmptyList &} 
	\fdepth{1 &}	1 &	[1] &	[[5]] &	[[4]] &	1/5 &	[4] &	[1] & 1 
	& 0.77    \\
   	& \bench{last} 		& 
   	\deps{\errorEmptyList &}
	\fdepth{1 &}	1 &	[2] &	[5*] &	[4*] &	2/4 &	4* &	[4,1] & 1
	&  1.04    \\
   	& \bench{init}		& 
   	\deps{\errorEmptyList &} 
    \fdepth{1 &}	1 &	[3] &	[5*] & [4*] &	2/8 &	[16,4] &	[16,1] & 1
    & 3.12
    \\
   	& \bench{fold1} 		& 
   	\deps{\errorEmptyList &}
    \fdepth{1 &}	1 &	[3] &	[5*] &	[4*] &	2/5 &	[4,16] &	[1,16] & 1
	& 2.41
    \\
   	& \bench{foldr1} 	& 
   	\deps{\errorEmptyList &}
    \fdepth{1 &}	1 &	[1] &	[[5]] &		[[4]] &	1/2 &	[4] &	[1] & 1
    & 1.08 \\
   	& \bench{(!!)}   	& 
   	\deps{\errorEmptyList &}
    \fdepth{1 &}	2 &	[4,4] &	[5*,10*] &	[4*,9*] &	4/9 &	36* &	[12,24,36,36] & 6
    &  7.81
    \\
   	& \bench{cycle}    	& 
   	\deps{\errorEmptyList &}
        \fdepth{1 &}	1 &	[2] &	[5*] &	[4*] &	2/6 &	4* &	[4,1] & 1 
	& 1.37
    \\
\hline
$3^{\text{rd}}$ 
   	& \bench{maximum}    & 
   	\deps{\errorEmptyList, \bench{fold1} &}
    \fdepth{1 &}	1 &	[3] &	[5*] &		[4*] &	2/4 &	[4,16] &	[1,16] &	1
    & 3.38
    \\
   	& \bench{minimum}    & 
   	\deps{\errorEmptyList, \bench{fold1} &}
	\fdepth{1 &}	1 &	[3] &	[5*] &	[4*] &	2/4 &	[4,16] &	[1,16] &	1    
	& 2.80
    \\
    \hline\hline
\multicolumn{11}{|l|}{\datalist (115 functions defined and verified)}\\
	    \hline
$1^{\text{st}}$  
	& \bench{maximumBy} & 
	\deps{- \bench{error} & } 
	\fdepth{1 &}	1 &	[3] &	[5*] &		[4*] &	2/5 &	[16,4] &	[16,1] &	1 & 
	2.24
    \\
	& \bench{minimumBy} & 
	\deps{- \bench{error} & }
	\fdepth{1 &}	1 &	[3] &	[5*] &		[4*] &	2/5 &	[16,4] &	[16,1] &	1
	& 2.40
    \\
	& \bench{transpose} & 
	\deps{- \bench{error} & }
	\fdepth{1 &}	1 &	[3] &	[12*] &		[5*] &	2/11 &	[25,5] &	[0,4] &	1 & 
	51.02
    \\
	& \bench{genIndex} & 
	\deps{- \bench{error} & }
	\fdepth{1 &}	2 &	[6,6] &	[2*,5*] &	[1*,4*] &	6/12 &	4* &	[3,4,4,1,1,4] &	0 & 1.83
    \\
    \hline\hline
\multicolumn{11}{|l|}{\listnonempty  (57 functions defined and verified)}\\
    \hline
$1^{\text{st}}$  
	& \bench{fromList} & 
	\deps{- &}
	\fdepth{1 &}	1 &	[2] &	[5*] &	[4*] &	2/11 &	4* &	[4,1] &	1  
	    & 1.41
    \\
	& \bench{cycle} & 
	\deps{- &}
    \fdepth{1 &}	1 &	[1] &	[[2]] &	[[1]] &	1/3 &	[1] &	[0] &	0 
    & 1.27
    \\
	& \ \ - \bench{toList} & 
	\deps{- &}
    \fdepth{1 &}	2 &	[1,2] &	[[2],5*] &		[[1],4*] &	2/3 &	4* &	[1,3] &	1
    & 3.11
    \\
	& \bench{(!!)} & 
	\deps{- &}
    \fdepth{1 &}	1 &	[4] &	[10*] &		[9*] &	4/22 &	9* &	[9,4,4,9] &	1
    & 5.96
    \\
    \hline
$2^{\text{nd}}$  
    & \bench{cycle} & 
    \deps{\bench{fromList} &}
    \fdepth{1 &}	2 &	[1,1] &	[[12],[5]] &		[[5],[1]] &	1/2 &	[5] &	[2] &	2
    & 3.13\\
    & \bench{lift} & 
    \deps{\bench{fromList} &}
    \fdepth{1 &}	1 &	[1] &	[[6]] &	[[5]] &	1/2 &	[5] &	[2] &	2
    & 2.90
    \\
    & \bench{inits} & 
    \deps{\bench{fromList}, \bench{DL.inits} &}
    \fdepth{1 &}	1 &	[1] &	[[6]] &		[[5]] &	1/5 &	[5] &	[2] &	2
    & 2.95 \\
    & \bench{tails} & 
    \deps{\bench{fromList}, \bench{DL.tails} & }
    \fdepth{1 &}	2 &	[1,1] &	[[5],[6]] &	[[4],[5]] &	1/5 &	[20] &	[6] &	6
    & 10.22
    \\
    & \bench{scanl} & 
    \deps{\bench{fromList}, \bench{L.scanl} & }
    \fdepth{1 &}	1 &	[1] &	[[6]] & [[5]] &	1/5 &	[5] &	[2] &	2
    & 2.23
    \\
    & \bench{scanl1} & 
    \deps{\bench{fromList}, \bench{L.scanl1} & }
    \fdepth{1 &}	1 &	[1] &	[[6]] &	[[5]] &	1/2 &	[5] &	[1] &	1
    & 1.13
    \\
    & \bench{insert} & 
    \deps{\bench{fromList}, \bench{DL.insert} & }
    \fdepth{1 &}	1 &	[1] &	[[12]] &	[[5]] &	1/5 &	[5] &	[2] &	2
    & 2.25 \\
    & \bench{transpose} & 
    \deps{\bench{fromList}, \bench{DL.transpose} & }
    \fdepth{1 &}	2 &	[1,1] &	[[7],[12]] &	[[6],[5]] &	1/7 &	[30] &	[6] &	6
    & 37.48
    \\
    \hline
$3^{\text{rd}}$  
    & \bench{reverse} & 
    \deps{\bench{lift}, \bench{L.reverse} & }
    \fdepth{1 &}	2 &	[1,1] &	[[5],[12]] &		[[4],[5]] &	2/3 &	[5,4] &	[2,3] & 5
    & 0.96
    \\
    & \bench{sort} & 
    \deps{\bench{lift}, \bench{DL.sort} &} 
    \fdepth{1 &}	2 &	[1,1] &	[[12],[5]]  &	[[5],[4]] &	2/3 &	[5,4] &	[2,3] &	5
    & 1.03 \\
    & \bench{sortBy} & 
    \deps{\bench{lift}, \bench{DL.sortBy} & }
    \fdepth{1 &}	2 &	[1,1] &	[[12],[5]] &	[[5],[4]] &	2/3 &	[5,4] &	[2,3] & 5 
    & 0.97
    \\
    \hline
\end{tabular}
\end{small}
\caption{Evaluation of Migrations Assistance.
\round: number of iterations to verify the function.
\funname: name of the function. 
\gradNum: number of $\egrad$ inserted. 
\occs:  times each $\egrad$ is used. 
For each occurrence, we give the number of template candidates (\cands)
and how many are specific (\precise).
\parts: the number of partitions. 
For each partition, we give
the number of 
concretizations (\instan) and safe concretizations (\sols).
\staticsols: number of static solutions found.
\ttime: time in sec.   
}
\label{table:evaluation}
\end{center}
\vspace{-6mm}
\end{table*}
\end{small}
 
\paragraph{Evaluation}

Table~\ref{table:evaluation} summarizes the migration case study;
there are three subtables, one for each library: \bench{GHC.List},
\bench{Data.List}, and \bench{Data.List.NonEmpty}.  Within each of these
tables, there is a row for every function which requires a refinement
type to type check.
Rows are collected into groupings of
rounds (\ie \textbf{Steps 1 - 4}), where round $i$ lead to refinement type preconditions 
that trigger type errors in the functions of round $i+1$.  
The columns of the table have
the same meaning as that of Table~\ref{table:indexing}; all results
are for depth 1.

\paragraph{\bench{GHC.List}:} 
\liquidHaskell reported three static errors on the original version of \bench{GHC.List}, 
\ie with no user refinement type specifications. 
The function \errorEmptyList is rejected as it is merely a wrapper around
the \bench{error} function; \bench{scanr}
and \bench{scanr1} each have incomplete patterns assuming a non-empty
list postcondition.
\toolname performed bad at all these three initial cases. 
It was unable to generate any \scshort for \errorEmptyList, 
since the required false precondition is non local. 
It required more than one hour to generate \scshorts of
\bench{scanr} and it timed-out in the case of \bench{scanr1}.
The reason for this is that @??@ in post-conditions
generated dependent set of constraints that made the partitioning 
optimization (\S~\ref{sec:implementation}) useless, despite its importance in the other scenarios.
Yet, \toolname was very efficient in the next two rounds. 
The specification of \errorEmptyList introduced errors in eight functions that were 
fixed using \toolname: in seven cases the generated \scshorts correspond to 
exactly one predicate that was used to replace the @??@.
One of these functions, \bench{fold1}, was further used by two more functions 
that were interactively migrated at the third and final specification round. 

In short, out of the 56 functions defined, 
13 required refinement type specifications, since the rest 
were typed with the true-default types generated by Liquid Haskell. 
\toolname was unusable in 2 cases, but synthesized exactly one, the correct, 
predicate in 10 of them. For the case of @(!!)@, \toolname generated 6
correct predicates, out of which we manually picked the most general one. 
The two cases where \toolname was unusable are due to many dependencies within the same partition; as future work, we should explore whether it is possible to devise a more advanced partitioning scheme.

\paragraph{\bench{Data.List}:} 
Migration of \bench{Data.List} only required one round. 
Four functions errored
due to incomplete patterns or violation of preconditions 
of functions imported from previously verified \bench{GHC.List}.
Unsurprisingly, \toolname was unable to find any \scshorts for \bench{genIndex}, 
a generic variant of \bench{(!!)} that indexes lists using any integral (instead of integer) index,
because it lacks arithmetic templates for integrals. 
To complete migration, the user needs to either 
manually provide the refinement type of \bench{genIndex}
or appropriately extend the set of templates. 

\paragraph{\bench{Data.List.NonEmpty}:}
Migration was more interesting for the  \bench{Data.List.NonEmpty} library that manipulates 
the data type @NonEmpty a@ of non-empty lists. 
The first round exposed that \bench{fromList} %
requires the non-empty precondition. 
The \bench{GHC.List} function \bench{cycle}
has a non-empty precondition, thus lifted to non-empty lists
does not type check.
\begin{mcode}
   cycle :: NonEmpty a -> NonEmpty a
   cycle = fromList . List.cycle . toList
\end{mcode}
To migrate \bench{cycle}, we first gradually refined the result type of 
@toList :: NonEmpty a -> {xs:[a] | ??}@ 
for which \toolname suggested the single static refinement of 
@0 < len xs@. 

Verification of non-empty list indexing calls requires
invariants that relate the lengths of the 
empty and non-empty lists. 
Similarly, \toolname finds \scshorts
only after predicate templates that
express such invariants are added.  
In general, to aid migration on user-defined structures,
\toolname requires the definition of domain-specific templates.

On the second round, 
the non-empty precondition of \bench{fromList}
triggers errors to eight clients (step 4a), 
all of which call functions that return (non-provably)
non-empty lists. 
For example, \bench{inits} lifts @List.inits@ 
to non-empty lists.
\begin{mcode}
  inits = fromList . List.inits . toList
\end{mcode}
To migrate such functions, we first assume an unknown specification, for example: 
\begin{mcode}
 assume List.inits :: [a] -> {o:[a] | ?? }
\end{mcode}
We then use \toolname to solve the unknown specification to @0 < len o@. 
At this point, we can either update the imported function with the discovered static specification and recheck the library (step 4c), or stick to gradual checking and stay with the unknown specification for @inits@.

The function @cycle@ reappears in the second round, due to the new 
precondition of @fromList@. 
Since the imported @List.cycle@ already has a precondition @{i:[a] | 0 < len i}@, 
at this round we further strengthened the existing precondition with the 
imprecise refinement: %
\begin{mcode}
 assume List.cycle :: i:{ [a] | 0 < len i && ?? } -> {o:[a] | ?? }
\end{mcode}
This is the only case in our experiments 
where we used a gradual refinement with a static part
and thus the only case in which some local 
templates are rejected as non-specific 
(here 3 out of 4 local templates are not specific).

Finally,
we use \toolname to derive higher-order specifications. 
The \bench{lift} function 
lifts a list transformation to non-empty lists,
@lift f = fromList . f . toList@,
and comes with a comment that 
``If the provided function returns an empty list, this will raise an error.''.
Alerted by this comment, we use a @??@ in higher-order position:
\begin{code}
  lift :: (i:[a] -> {o:[b] | ??}) -> NonEmpty a -> NonEmpty b
\end{code}
\toolname
produces two static solutions @0 < len o@
and @len i == len o@. We choose the second, which leads 
to type errors in three clients, which are resolved in
later rounds.

To sum up,
\toolname indeed is aiding Haskell to Liquid Haskell 
migration of real libraries, since the user 
can place the initial @??@s and choose from the suggested \scshorts, 
instead of writing specifications from scratch. 
Often times, there is only one possible \scshort
coinciding to all concretizations, thus the choice is trivial. 
When no suggestions are generated, 
\eg @errorEmp@ or @genIndex@, 
the user can fall back to the standard verification process. 

Automatic insertion of @??@s would further reduce the required user input, but is left as future work. 
Theorem~\ref{thm:embedding} suggests a complete way to automatically insert @??@s so that the migrating libraries gradually type check. 
The types suggested by this theorem are very imprecise, yet they provide 
a sound starting point, upon which one can increase precision so long as 
the libraries type check.

\section{Related Work}
\label{sec:related}

\paragraph{Liquid Types.} 
Dependent types 
allow arbitrary expressions at the type level to express theorems on programs, 
while theorem proving is simplified by various automations 
ranging from tactics (\eg Coq~\citep{coq-book}, Isabelle~\citep{IsabelleManual})
to external SMT solvers (\eg F*~\citep{fstar}). 
Liquid types~\citep{LiquidPLDI08} restrict the expressiveness of 
the type specifications to decidable fragments of logic 
to achieve decidable type checking and inference. 

\paragraph{Gradual Refinement Types.}
Several refinement type systems mix 
static verification with runtime checking.
Hybrid types~\citep{Knowles10} use an external prover
to statically verify subtyping when possible, 
otherwise a cast is automatically inserted to defer checking at runtime. 
Soft contract verification~\citep{HOSE12,SCV14,Nguyen18} works in the
other direction, statically verifying contracts wherever
possible, and otherwise leaving unverified contracts for checking at
runtime.
\citet{Ou2004} allow the programmer to explicitly annotate 
whether an assertion is verified 
at compile- or runtime. 
Manifest contracts~\citep{Greenberg10}
formalize the metatheory of refinement typing
in the presence of dynamic contract checking.
\citet{Lehmann17} developed the first gradual refinement type system, which adheres to the refined criteria of~\citet{siek15}.
None of these systems support inference; 
on the contrary, because refinements can be arbitrary, inference is impossible in these systems.
Here we restrict gradual refinements 
to a finite set of predicates, in order to achieve
inference by adaptation of the liquid inference procedure. 

\paragraph{Gradual Type Inference.}
Several approaches have been developed to combine type inference and gradual types. 
\citet{siek08} infer gradual types using unification, 
while~\citet{RastogiCH12}
exploit type inference to improve the performance of gradually-typed programs.
\citet{Garcia2015}
lift an inference algorithm from a core system to its 
gradual counterpart, by ignoring the unification constraints imposed by gradual types.
In contrast, in gradual liquid inference
the constraints imposed by gradual refinements cannot be ignored, since  
unlike Hindley-Milner inference, 
the liquid algorithm starts from the strongest solution 
for the liquid variables 
and uses constraints to iteratively weaken the solution. Our work is the first gradual inference work to be systematically derived from the static algorithm based on the Abstracting Gradual Typing approach (AGT)~\citep{adt}, and proven to satisfy the static gradual guarantee~\citep{siek15}.

\paragraph{Error Reporting.} Properly localizing, diagnosing and reporting type errors is a long-standing challenge, especially for inference algorithms. Since the seminal work of \citet{Wand86}, which keeps track of all unification steps in order to help debugging, many algorithms have been proposed to better assist programmers. A large variety of techniques has been explored, including
slicing~\citep{TipDineshTOSEM01,HaackWellsESOP03}, heuristics~\citep{Zhang15}, SMT constraint solving~\citep{PavlinovicOOPSLA14}, counter-example generation~\citep{NguyenH15,SeidelICFP16}, machine learning~\citep{SeidelSCWJ17}, and other search-based approaches, such as Seminal~\citep{LernerPLDI07} and counterfactual change inference~\citep{ChenErwigJFP18}.

Apart from~\citet{SeidelICFP16}, which focuses on identifying dynamic witnesses for static refinement type errors, none of the above approaches target refinement types. In this work, we uncover a novel application of gradual typing for error explanation, by observing that the notion of concretizations of gradual types that stem from AGT, embed useful justifications (or lack thereof) that can be exploited to identify refinement errors and guide migration.

\paragraph{Gradual Program Migration.} Recently, \citet{CamporaPOPL18} exploit variational typing~\citep{ChenTOPLAS14} to assist programmers in making their program more static, by analyzing the impact of replacing unknown types with static types. Variational typing greatly reduces the complexity of exploring all possible combinations. Though the considered type system is simple, it seems likely that the approach could be extended to refinements and combined with our technique. This might be an effective way to address the scalability issues we have encountered with our implementation approach in certain scenarios.

\section{Conclusion}
\label{sec:conclusion}

This paper makes the novel observation that gradual inference based on abstract interpretation can be fruitfully exploited to assist in explaining type errors and migrating programs to a stronger typing discipline. 
We develop this intuition in the context of refinement
types, yielding a novel integration of liquid type inference and
gradual refinements. Gradual liquid type inference computes possible
concretizations of unknown refinements in order for a program to be
well-typed. In addition to laying down the theoretical foundations of gradual liquid type inference, we provide an implementation integrated with \liquidHaskell. Thanks to a number of heuristics and optimizations, the current implementation of the interactive tool \toolname proves useful for migrating existing Haskell libraries to the stronger discipline of \liquidHaskell. Our experience however also shows that enhancing the scalability of our prototype further is necessary; we believe that variational typing could be particularly helpful towards that goal. Similarly, integrating prior work on ranking suggested type error sources would help our inference algorithm suggest the most relevant concretizations first. 
Finally, we conjecture that the idea of using gradual typing for error reporting and migration generalizes to other typing disciplines.

\clearpage
{
\bibliography{sw}
\balance
}

\clearpage
\appendix

\section{Auxiliary Definitions and Proofs}
\label{appendix:section:metatheory}

We now provide the full inference algorithm 
and prove the theorems of~\S~\ref{subsec:metatheory}.

\subsection{The Inference Algorithm}\label{subsec:inference-algorithm}

We define $\ginfername$ as in~\S~\ref{subsec:metatheory}: 
\begin{mcode}
$\ginfername$ :: $\glenvtyp$ -> $\glexprtyp$ -> Quals -> [$\gltypetyp$] 
$\ginfer{\glenv}{\glexpr}{\iquals}$ = { $\gltyp'$ | Just $\gltyp'$ <- $\maybeapplysub{\gltyp}{\ltsol}$ 
                    , $\ltsol \in \qgsolve{\ltsol_0}{\gliquid{C}}$ }
  where 
    ltsolzero = $\lambda \kvar$. iquals
    ($\gltyp$, $\gliquid{C}$) = Cons $\glenv$ $\glexpr$ 
    
$\gsolvename$ :: Sol -> [$\glconstyp$] -> Quals -> [Maybe Sol]
$\qgsolve{\ltsol}{\gliquid{C}}$ = $\{\solve{\ltsol}{\liquid{C}} \spmid \algradualinstance{C}\}$ 
\end{mcode}
Note that unlike~\citep{LiquidPLDI08}, for simplicity, 
we assume that the 
set of refinement predicates $\iquals$ is not 
a set of templates, but an ``instantiated'' set of predicates.
The algorithmic concretization on a list of constraints is defined as follows
\[ \arraycolsep=1pt 
\begin{array}{rcl}
\lconcrete{\bpred{\glpred}{\rbind}} &\defeq&
  \powerset{\iquals} \cap \concrete{\bpred{\glpred}{\rbind}}\\
\lconcrete{\gliquid{\env}}
&\defeq& \{\lenv \mid 
\bind{x}{\ltyp}\in\lenv \ \textit{iff}\
\bind{x}{\gliquid{\typ}}\in\gliquid{\env},
\algradualinstance{\typ} \}\\
\lconcrete{\issubtype{}{\glenv}{\gltyp_1}{\gltyp_2}}
&\defeq& \{\issubtype{}{\lenv}{\ltyp_1}{\ltyp_2} \mid 
\algradualinstance{\env},
\algradualinstance{\typ_i} \}\\
\lconcrete{\iswellformed{}{\glenv}{\gltyp}}
&\defeq& \{\iswellformed{}{\lenv}{\ltyp} \mid 
\algradualinstance{\env},
\algradualinstance{\typ}, \}\\
\lconcrete{\gliquid{C}}
&\defeq& \{\liquid{C} \mid 
c\in \liquid{C}\ \textit{iff}\
\gliquid{c}\in\gliquid{C},
\algradualinstance{c} \}
\end{array}
\]

We complete the definitions by repeating the 
@Solve@ and @Cons@ algorithms from~\citep{LiquidPLDI08}.

The procedure @Solve@ $\ltsol$ $\liquid{C}$
repeatedly weakens the solution $\ltsol$
until the set of constraints $\liquid{C}$ is satisfied, 
or returns @Nothing@.
\begin{mcode}
Solve :: Sol -> [$\lconstyp$] -> Maybe Sol
Solve $\ltsol$ $\liquid{C}$ = 
  if exists $\liquid{c}$ $\in$ $\liquid{C}$ s.t. not $\valid{\applysub{\liquid{c}}{\ltsol}}$
  then case Weaken $\liquid{c}$ $\ltsol$ of 
        Just $\ltsol'$ -> Solve $\ltsol'$ $\liquid{C}$ 
        Nothing -> Nothing
  else Just $\ltsol$

Weaken :: $\lconstyp$ -> Sol -> Maybe Sol
Weaken ($\iswellformed{}{\ltenv}{\tref{v}{\btyp}{\applysub{\kvar}{\theta}}}$) $\ltsol$ = Just $\$$ 
  A[$\kvar$ $\mapsto$ {q|q $\in$ $\applysub{\kvar}{\ltsol}$, $\valid{\iswellformed{}{\applysub{\ltenv}{\ltsol}}{\tref{v}{\btyp}{\applysub{q}{\theta}}}}$}]
Weaken $\issubtype{}{\ltenv}{\tref{v}{\btyp}{\pred}}{\tref{v}{\btyp}{\applysub{\kappa}{\theta}}}$ $\ltsol$ = Just $\$$ 
  A[$\kvar$ $\mapsto$ {q|q $\in$ $\applysub{\kvar}{\ltsol}$, $\valid{\issubtype{}{\applysub{\ltenv}{\ltsol}}{\tref{v}{\btyp}{\applysub{\pred}{\ltsol}}}{\tref{v}{\btyp}{\applysub{q}{\theta}}}}$}]
Weaken _ _  = Nothing
\end{mcode}

The procedure @Cons@ $\glenv$ $\glexpr$ is following the rules of Figure~\ref{fig:rules}
to generate a template type of the expression $\glexpr$ and the set of 
constraints that should be satisfied. 
\begin{mcode}
Cons :: $\glenvtyp$ -> $\glexprtyp$ -> (Maybe $\gltypetyp$, [$\glconstyp$])
Cons $\glenv$ $\glexpr$ = 
  let ($\gltyp$,$\glC$) = Gen $\glenv$ $\glexpr$
  ($\gltyp$, Split $\glC$)
\end{mcode}

The procedure @Gen @ $\glenv$ $\glexpr$ is using the typing rules of Figure~\ref{fig:rules}
to generate a template type and a set of constraints.
\begin{mcode}
Gen :: $\glenvtyp$ -> $\glexprtyp$ -> (Maybe $\gltypetyp$, [$\glconstyp$])
Gen $\glenv$ $x$ 
  = if $\glenv(x)$ = $\tref{v}{\btyp}{\_}$ 
     then (Just $\gtref{v}{\btyp}{v = x}$, $\emptyset$) 
     else (Just $\glenv(x)$, $\emptyset$)
Gen $\glenv$ $c$  
  = (Just $\tc{c}$, $\emptyset$) 
Gen $\glenv$ ($\ecast{\glexpr}{\gltyp}$) =
  let (Just $\gltyp_e$, $\glC$) = Gen $\glenv$ $\glexpr$ in 
  (Just $\gltyp$, ($\issubtype{}{\glenv}{\gltyp_e}{\gltyp}$, $\iswellformed{}{\glenv}{\gltyp}$, $\glC$)) 
Gen $\glenv$ ($\efun{x}{\gltyp_x}{\glexpr}$) = 
  let Just $\tfun{x}{\gltyp_x}{\gltyp}$ = Fresh $\glenv$ $\efun{x}{\gtyp_x}{\glexpr}$ in 
  let (Just $\gltyp_e$,$\glC$) = Gen ($\glenv,\bind{x}{\gltyp_x}$) $\glexpr$ in 
  (Just ($\tfun{x}{\gltyp_x}{\gltyp}$), ($\issubtype{}{\glenv}{\gltyp_e}{\gltyp}$, $\iswellformed{}{\glenv}{\tfun{x}{\gltyp_x}{\gltyp}}$,$\glC$))
Gen $\glenv$ ($\eapp{\glexpr}{y}$) = 
  let (Just ($\tfun{x}{\gltyp_x}{\gltyp}$), $\glC_1$) = Gen $\glenv$ $\glexpr$ in 
  let (Just $\gltyp_y$,$\glC_2$) = Gen $\glenv$ $y$ in 
  (Just $\gltyp\sub{x}{y}$,($\issubtype{}{\glenv}{\gltyp_x}{\gltyp_y}$, $\glC_1$ $\cup$ $\glC_2$))
Gen $\glenv$ $\eif{x}{\glexpr_1}{\glexpr_2}$) = 
  let Just $\gltyp$ = Fresh $\glenv$ $\eif{x}{\glexpr_1}{\glexpr_2}$ in 
  let (Just $\gltyp_1$,$\glC_1$) = Gen $\glenv,\bind{\_}{\tref{v}{\tbool}{\refparam{\rparam}{x}}}$ $\lexpr_1$ in 
  let (Just $\gltyp_2$,$\glC_2$) = Gen ($\glenv,\bind{\_}{\tref{v}{\tbool}{\refparam{\rparam}{\lnot x}}}$) $\glexpr_2$ in 
  (Just $\gltyp$, ($\iswellformed{}{\glenv}{\gltyp}$, $\issubtype{}{\glenv}{\gltyp_1}{\gltyp}$, $\issubtype{}{\glenv}{\gltyp_2}{\gltyp}$,$\glC_1$ $\cup$ $\glC_2$))
Gen $\glenv$ ($\elet{x}{\gltyp_x}{\glexpr_x}{\glexpr}$) = 
  let Just $\gltyp$ = Fresh $\glenv$ ($\elet{x}{\gltyp_x}{\glexpr_x}{\glexpr}$) in 
  let (Just $\gltyp_x$,$\glC_x$) = Gen $\glenv$ $\glexpr_x$ in 
  let (Just $\gltyp_e$,$\glC_e$) = Gen ($\glenv,\bind{x}{\gltyp}$) $\glexpr$ in 
  (Just $\gltyp$, ($\iswellformed{}{\glenv}{\gltyp}$, $\issubtype{}{\glenv}{\gltyp_e}{\gltyp}$, $\glC_x$ $\cup$ $\glC_e$))
Gen $\glenv$ ($\elettyp{x}{\gltyp_x}{\glexpr_x}{\glexpr}$) = 
  let Just $\gltyp$ = Fresh $\glenv$ ($\elet{x}{\gltyp_x}{\glexpr_x}{\glexpr}$) in 
  let (Just $\gltyp_e$,$\glC_e$) = Gen ($\glenv,\bind{x}{\gltyp}$) $\glexpr$ in 
  (Just $\gltyp$, ($\iswellformed{}{\glenv}{\gltyp_x}$, $\iswellformed{}{\glenv}{\gltyp}$, $\issubtype{}{\glenv}{\gltyp_e}{\gltyp}$, $\glC_x$ $\cup$ $\glC_e$))
Gen _ _ = 
  (Nothing, $\emptyset$)
  
    
Fresh :: $\glenvtyp$ -> $\glexprtyp$ -> Maybe $\gltypetyp$
Fresh $\glenv$ $\glexpr$ = HM type inference 
\end{mcode}
The procedure @Gen@ is using @Fresh@, the Hyndler Miller, unrefined type inference algorithm 
to generate liquid type templates with fresh refinement variables for the unknown types. 

Finally, @Split C@ using the well-formedness and sub-typing rules of Figure~\ref{fig:rules}
to split the constraints into basic constraints.
\begin{mcode}
Split :: [$\glconstyp$]->  [$\glconstyp$]
Split $\emptyset$ 
  = $\emptyset$
Split ($\iswellformed{}{\glenv}{\tref{v}{\btyp}{\pred}}$,C) 
  = ($\iswellformed{}{\glenv}{\tref{v}{\btyp}{\pred}}$, Split C)
Split ($\iswellformed{}{\glenv}{\tfun{x}{\gltyp_x}{\gltyp}}$,C) 
  = Split ($\iswellformed{}{\glenv}{\gltyp_x}$, $\iswellformed{}{\glenv,\bind{x}{\gltyp_x}}{\gltyp}$,C)
Split ($\issubtype{}{\glenv}{\tref{v}{\btyp}{\pred_1}}{\tref{v}{\btyp}{\pred_2}}$, C) 
  = ($\issubtype{}{\glenv}{\tref{v}{\btyp}{\pred_1}}{\tref{v}{\btyp}{\pred_2}}$, Split C)
Split ($\issubtype{}{\glenv}{\tfun{x}{\gltyp_{x1}}{\gltyp_1}}{\tfun{x}{\gltyp_{x2}}{\gltyp_2}}$, C) 
  = Split ($\issubtype{}{\glenv}{\gltyp_2}{\gltyp_1}$, $\issubtype{}{\glenv, \bind{x}{\gltyp_{x2}}}{\gltyp_1}{\gltyp_2}$, C)
\end{mcode}

\subsection{Correctness of Inference}\label{subsec:correctness}
Next, we prove Theorem~\ref{theorem:correctness}.
Let $\iquals$ be a finite set of predicates  from SMT-decidable logic, 
\glenv be a gradual liquid environment, and
\glexpr be a gradual liquid expression. 

The proofs rely on the properties of the functions 
@Solve@ and @Cons@. Since these two functions 
operate on liquid types, and are ignorant of the gradual setting, 
we directly port the proofs from~\citep{rondonthesis}. 

\begin{lemma}[Constraint Generation]\label{theorem:constraints}
Let $(\texttt{Just}\ \gltyp, \glC)$ $=$ @Cons@ \glenv \glexpr. 
\hastype{}{\glenv}{\glexpr}{\gltyp'}
\textit{iff}
there exists $\ltsol$ so that 
$\gltyp' \equiv \applysub{\gltyp}{\ltsol}$ and
\gradualvalid{\applysub{\glC}{\ltsol}}.
\end{lemma}
\begin{proof}
Following Theorem 4 of Appendix A of~\citep{rondonthesis}. 
Since @Cons@ is just the algorithmic version of the rules of Figure~\ref{fig:rules}
the theorem holds for any refinement type system with refinement variables, 
when the \valid{\cdot} relation is exactly the same as in the premises of the rules in 
Figure~\ref{fig:rules}.
\end{proof}

\begin{lemma}[Constraint Solving]\label{lemma:solve}
For every set of constraints $\lC$ and qualifiers $\quals$, 
\begin{enumerate}[leftmargin=*]
\item \solve{(\lambda \kappa. \quals)}{\lC} terminates. 
\item If $\solve{(\lambda \kappa. \quals)}{\lC}= \texttt{Just}\ \ltsol$ then 
\valid{\applysub{\lC}{\ltsol}}. 
\item If $\solve{(\lambda \kappa. \quals)}{\lC} = \texttt{Nothing}$
then $\lC$ has no solution on $\quals$.
\end{enumerate}
\end{lemma}
\begin{proof}
Theorem 6 of Appendix A of~\citep{rondonthesis}.
\end{proof}

\begin{lemma}[Gradual Validity]\label{lemma:validity} \ 
\begin{enumerate}[leftmargin=*]
\item
\gradualvalid{\applysub{\gliquid{c}}{\ltsol}}
\textit{iff}
$\exists \algradualinstance{c}. \valid{\applysub{\gliquid{c}}{\ltsol}}$
\item
\gradualvalid{\applysub{\gliquid{C}}{\ltsol}}
\textit{iff}
$\exists \algradualinstance{C}. \valid{\applysub{\gliquid{C}}{\ltsol}}$
\end{enumerate}
\end{lemma}
\begin{proof}\
\begin{enumerate}[leftmargin=*]
\item By case analysis on the shape of the constraint:
\begin{itemize}[leftmargin=*]
\item 
$\gliquid{c} \equiv \issubtype{}{\glenv}{\gltyp_1}{\gltyp_2}$. 
$$
\begin{array}{c}
\gradualvalid{\issubtype{}{\applysub{\glenv}{\ltsol}}{\applysub{\gltyp_1}{\ltsol}}{\applysub{\gltyp_2}{\ltsol}}} \\
\Leftrightarrow\\
\exists
\lenv\in\lconcrete{\glenv},
\ltyp_i\in\lconcrete{\gltyp_i}.
\valid{\issubtype{}{\applysub{\lenv}{\ltsol}}{\applysub{\ltyp_1}{\ltsol}}{\applysub{\ltyp_2}{\ltsol}}}\\
\Leftrightarrow\\
\exists
\algradualinstance{c}.
\valid{\applysub{\liquid{c}}{\ltsol}}\\
\end{array}
$$
\item $\gliquid{c} \equiv \iswellformed{}{\genv}{\gtyp}$.
$$
\begin{array}{c}
\gradualvalid{\iswellformed{}{\applysub{\glenv}{\ltsol}}{\applysub{\gltyp}{\ltsol}}}\\
\Leftrightarrow\\
\exists
\lenv\in\lconcrete{\glenv},
\ltyp\in\lconcrete{\gltyp}.
\valid{\iswellformed{}{\applysub{\lenv}{\ltsol}}{\applysub{\ltyp}{\ltsol}}}\\
\Leftrightarrow\\
\exists
\algradualinstance{c}.
\valid{\applysub{\liquid{c}}{\ltsol}}\\
\end{array}
$$
\end{itemize}
\item By the definition of \gradualvalid{\cdot}
and concretization of list of constraints. 
$$
\begin{array}{c}
\gradualvalid{\applysub{\gliquid{C}}{\ltsol}}\\
\Leftrightarrow\\
\forall \gliquid{c}\in\gliquid{C}.
\gradualvalid{\applysub{\gliquid{c}}{\ltsol}}\\
\Leftrightarrow\\
\forall \gliquid{c}\in\gliquid{C}.
\exists \algradualinstance{c}. \valid{\applysub{\gliquid{c}}{\ltsol}}\\
\Leftrightarrow\\
\exists \algradualinstance{C}. \valid{\applysub{\gliquid{C}}{\ltsol}}
\end{array}
$$
\end{enumerate}
\end{proof}

\begin{theorem}[Soundness]\ \\
If $\gltyp \in \ginfer{\glenv}{\glexpr}{\iquals}$, 
then \hastype{}{\glenv}{\glexpr}{\gltyp}. 
\end{theorem}
\begin{proof}
Since $\gltyp \in \ginfer{\glenv}{\glexpr}{\iquals}$
then $\exists \ltsol$ so that

\noindent
$
\begin{array}{rl}
&\\
\quad(1)&\gltyp = \applysub{\gltyp'}{\ltsol}\hfill \\
(2)&\texttt{Just}\ \ltsol \in \qgsolve{(\lambda \kappa. \quals)}{\gliquid{C}} \\
(3)&(\texttt{Just}\ \gltyp',\glC) = \texttt{Cons}\ \glenv\ \glexpr \\
&\\
\multicolumn{2}{l}{\text{From}\ (2),\ \exists \algradualinstance{C}\ \text{so that}}\\
&\\
(4)&\texttt{Just}\ \ltsol \in \solve{(\lambda \kappa. \quals)}{\lC} \\
&\\
\multicolumn{2}{l}{\text{From}\ (4)\ \text{and Lemma~\ref{lemma:solve} we get}}\\
&\\
(5)&\valid{\applysub{\lC}{\ltsol}}\\
&\\
\multicolumn{2}{l}{\text{By Lemma~\ref{lemma:validity} we get}}\\
&\\
(6)&\gradualvalid{\applysub{\glC}{\ltsol}}\\
&\\
\multicolumn{2}{l}{\text{By}\ (1),\ (3),\ \text{and Theorem~\ref{theorem:constraints}},}\\
&\\
&\hastype{}{\glenv}{\glexpr}{\gltyp}\\
&\\
\end{array}
$

\end{proof}

\begin{theorem}[Completeness]\label{theorem:completeness}\ \\
If $\ginfer{\glenv}{\glexpr}{\iquals} = \emptyset$, 
then $ \not \exists\gltyp.\
\hastype{}{\glenv}{\glexpr}{\gltyp}$.
\end{theorem}
\begin{proof}
Assume that there exists $\gltyp$ so that
$\hastype{}{\glenv}{\glexpr}{\gltyp}$. 
Then, by Lemma~\ref{theorem:constraints}, there exists an \ltsol  so that

\noindent
$
\begin{array}{rl}
\quad(1)&(\texttt{Just}\ \gltyp',\glC) = \texttt{Cons}\ \glenv\ \glexpr \\
(2)&\gltyp = \applysub{\gltyp'}{\ltsol}\hfill \\
(3)&\gradualvalid{\applysub{\glC}{\ltsol}} \\
&\\
\multicolumn{2}{l}{\text{From}\ (3)\ \text{and Lemma~\ref{lemma:validity}}\ 
\exists \algradualinstance{C}\ \text{so that}}\\
&\\
(4)&\valid{\applysub{\lC}{\ltsol}} \\
&\\
\multicolumn{2}{l}{\text{From}\ (4)\ \text{and inverting 3 of Lemma~\ref{lemma:solve} we get}}\\
&\\
(5)&\solve{(\lambda \kappa. \quals)}{\lC}\not = \texttt{Nothing}  \\
&\\
\multicolumn{2}{l}{\text{By the definition of \gsolvename we get}}\\
&\\
(6)&\exists\ltsol'.\texttt{Just}\ \ltsol' \in \qgsolve{(\lambda \kappa. \quals)}{\gliquid{C}}\\
&\\
\multicolumn{2}{l}{\text{By the definition of \ginfername we get}}\\
&\\
&\ginfer{\glenv}{\glexpr}{\iquals} \not = \emptyset\\
\end{array}
$

Since we reached a contradiction,  
there cannot exist $\gltyp$ so that
$\hastype{}{\glenv}{\glexpr}{\gltyp}$. 
\end{proof}

\begin{theorem}[Termination]
\ginfer{\glenv}{\glexpr}{\iquals} terminates.
\end{theorem}
\begin{proof}
Since both the set of constraints $\glC$ and 
the set of refinement predicates $\iquals$ are finite, 
the concretizations $\lconcrete{\glC}$ are also finite. 
Thus, \ginfer{\glenv}{\glexpr}{\iquals} calls 
\solve{\cdot}{\cdot} finite times and by Theorem~\ref{lemma:solve}
\solve{\cdot}{\cdot} terminates, 
thus so does \ginfer{\glenv}{\glexpr}{\iquals}.
\end{proof}

\subsection{Criteria for Gradual Typing}\label{subsec:gradual:criteria}
Finally we prove the static criteria for gradual typing.

\paragraph{(i) Conservative Extension}

\begin{theorem}[Conservative Extension]
If $\texttt{Just}\ \ltyp = \infer{\lenv}{\lexpr}{\iquals}$, 
then $\ginfer{\lenv}{\lexpr}{\iquals} = \{\ltyp\}$. 
Otherwise, 
$\ginfer{\lenv}{\lexpr}{\iquals} = \emptyset$. 
\end{theorem}
\begin{proof}
If $\texttt{Just}\ \ltyp = \infer{\lenv}{\lexpr}{\iquals}$, 
then there exists an $\ltsol$, so that

\noindent
$
\begin{array}{rl}
\quad(1)&\ltyp = \applysub{\ltyp'}{\ltsol}\hfill \\
(2)&\texttt{Just}\ \ltsol = \solve{(\lambda \kappa. \quals)}{\liquid{C}} \\
(3)&(\texttt{Just}\ \ltyp',\lC) = \texttt{Cons}\ \lenv\ \lexpr \\
\end{array}
$
Since the generated constraints $\lC$ contain no \egrad, 
then 
$\{ \lC\} = \lconcrete{\lC}$.

Thus, 
$\qgsolve{(\lambda \kappa. \quals)}{\lC} = \texttt{Just}\ \ltsol$.
So, $\ginfer{\lenv}{\lexpr}{\iquals} = \{\ltyp\}$.

Otherwise, $\infer{\lenv}{\lexpr}{\iquals} = \texttt{Nothing}$,
because of a failure either at constraint generation 
or at solving. In either case \ginfername will also return $\emptyset$.
\end{proof}

\paragraph{(ii) Embedding of Unrefined Terms}

\begin{definition}[Unrefined Type \& Terms]
Unrefined types and terms represent base types and lambda calculus terms
that are typed using HindleyMilner inference: 
$\hastype{}{\shape{\env}}{\shape{\expr}}{\shape{\typ}}$.

\begin{align*}
\shape{\tref{v}{\btyp}{\pred}} &= \btyp \\
\shape{\tfun{x}{\typ}{\typ}} &= \shape{\typ_x} \rightarrow\shape{\typ}
\end{align*}

\begin{align*}
\shape{c} &= c \\
\shape{\efun{x}{\typ}{\expr}} &= \efun{x}{\typ}{\shape{\expr}}\\
\shape{x} &=x\\
\shape{\eapp{\expr}{x}} &= \eapp{\shape{\expr}}{x}\\
\shape{\eif{x}{\expr_1}{\expr_2}} &= \eif{x}{\shape{\expr_1}}{\shape{\expr_2}}\\
\shape{\elet{x}{\typ}{\expr_x}{\expr}} &= \elet{x}{\shape{\typ}}{\shape{\expr_x}}{\shape{\expr}}\\
\shape{\elettyp{x}{\typ}{\expr_x}{\expr}} &= \elettyp{x}{\shape{\typ}}{\shape{\expr_x}}{\shape{\expr}}\\
\end{align*}
\end{definition}

\begin{definition}[Imprecise Types \& Terms]
Imprecise types are refined with only \egrad.
Imprecise terms only use imprecise type annotations.
\begin{align*}
\dyn{\tref{v}{\btyp}{\pred}} &= \tref{v}{\btyp}{\egrad} \\
\dyn{\tfun{x}{\typ}{\typ}} &= \tfun{x}{\dyn{\typ_x}}{\dyn{\typ}}
\end{align*}

\begin{align*}
\dyn{c} &= c' \text{where } c' = c \wedge \tc{c'}=\dyn{\tc{c}} \\
\dyn{\efun{x}{\typ}{\expr}} &= \efun{x}{\typ}{\dyn{\expr}}\\
\dyn{x} &=x\\
\dyn{\eapp{\expr}{x}} &= \eapp{\dyn{\expr}}{x}\\
\dyn{\eif{x}{\expr_1}{\expr_2}} &= \eif{x}{\dyn{\expr_1}}{\dyn{\expr_2}}\\
\dyn{\elet{x}{\typ}{\expr_x}{\expr}} &= \elet{x}{\dyn{\typ}}{\dyn{\expr_x}}{\dyn{\expr}}\\
\dyn{\elettyp{x}{\typ}{\expr_x}{\expr}} &= \elettyp{x}{\dyn{\typ}}{\dyn{\expr_x}}{\dyn{\expr}}\\
\end{align*}
\end{definition}

\begin{lemma}[Well formedness of imprecise types]\label{lemma:dynwellformed}
\iswellformed{}{\env}{\dyn{\typ}}
\end{lemma}
\begin{proof}
Trivial, since true always belongs to the concretization of imprecise gradual refinements. 
\end{proof}

\begin{lemma}[Imprecise Subtyping]\label{lemma:dynsubtype}
If all of the refinements in $\typ$ are local, then 
\issubtype{}{\env}{\typ}{\dyn{\typ}} and 
\issubtype{}{\env}{\dyn{\typ}}{\typ}.
\end{lemma}
\begin{proof}
By induction on $\typ$.
\begin{itemize}[leftmargin=*]
\item \issubtype{}{\env}{\tref{v}{\btyp}{\pred}}{\tref{v}{\btyp}{\egrad}}, since
$\tref{v}{\btyp}{\etrue}\in\concrete{\tref{v}{\btyp}{\egrad}}$.
\item \issubtype{}{\env}{\tref{v}{\btyp}{\egrad}}{\tref{v}{\btyp}{\pred}}, since
$\tref{v}{\btyp}{\pred}\in\concrete{\tref{v}{\btyp}{\egrad}}$.
\item \issubtype{}{\env}{\tfun{x}{\typ_x}{\typ}}{\tfun{x}{\dyn{\typ_x}}{\dyn{\typ}}}, since
\issubtype{}{\env}{\dyn{\typ_x}}{\typ_x} and 
\issubtype{}{\env, \bind{x}{\dyn{\typ_x}}}{\typ}{\dyn{\typ}}
by inductive hypothesis. 
\item \issubtype{}{\env}{\tfun{x}{\dyn{\typ_x}}{\dyn{\typ}}}{\tfun{x}{\typ_x}{\typ}}, since
\issubtype{}{\env}{\typ_x}{\dyn{\typ_x}} and 
\issubtype{}{\env, \bind{x}{\typ_x}}{\dyn{\typ}}{\typ}
by inductive hypothesis. 
\end{itemize}
\end{proof}

\begin{lemma}[Imprecise Terms]\label{lemma:dynamic}
If all the refinements for constants and user types are local, then
if $\hastype{}{\shape{\env}}{\shape{\expr}}{\shape{\typ}}$, 
then $\hastype{}{\dyn{\env}}{\dyn{\expr}}{\dyn{\typ}}$.
\end{lemma}

\begin{proof}
The proof proceeds by induction on the derivation tree of 
$\hastype{}{\shape{\env}}{\shape{\expr}}{\shape{\typ}}$.
\begin{itemize}[leftmargin=*]
\item $\expr \equiv x$.
By assumption, $x\in\dyn{\env}$. 
\begin{itemize}[leftmargin=*]
\item If $\dyn{\env}(x) = \tref{v}{\btyp}{\_}$, 
then $\hastype{}{\dyn{\env}}{x}{\tref{v}{\btyp}{v = x}}$. 
By Rule~\ruletsub and Lemmas~\ref{lemma:dynsubtype} and~\ref{lemma:dynwellformed}
 $\hastype{}{\dyn{\env}}{x}{\tref{v}{\btyp}{\egrad}}$.
\item Otherwise, $\hastype{}{\dyn{\env}}{x}{\dyn{\env(x)}}$. 
\end{itemize}

\item $\expr \equiv \efun{x}{\typ}{\expr'}$. 
By inversion of hypothesis 
$\hastype{}{\shape{\env, \bind{x}{\typ_x}}}{\shape{\expr'}}{\shape{\typ}}$. 
By inductive hypothesis
$\hastype{}{\dyn{\env, \bind{x}{\typ_x}}}{\dyn{\expr'}}{\dyn{\typ}}$. 
By Lemma~\ref{lemma:dynwellformed}
$\iswellformed{}{\dyn{\env}}{\dyn{\tfun{x}{\typ_x}{\typ}}}$. 
So, by rule \ruletfun
$\hastype{}{\dyn{\env}}{\dyn{\expr}}{\dyn{\tfun{x}{\typ_x}{\typ}}}$.
\item $\expr \equiv \eapp{\expr'}{x}$.
By inversion of hypothesis 
$\hastype{}{\shape{\env}}{\shape{\expr'}}{\shape{\tfun{x}{\typ_x}{\typ}}}$
and
$\hastype{}{\shape{\env}}{\shape{x}}{\shape{\typ_x}}$.
By inductive hypothesis
$\hastype{}{\dyn{\env}}{\dyn{\expr'}}{\dyn{\tfun{x}{\typ_x}{\typ}}}$
and
$\hastype{}{\dyn{\env}}{\dyn{x}}{\dyn{\typ_x}}$.
By rule \ruletapp
$\hastype{}{\dyn{\env}}{\dyn{\expr}}{\dyn{\typ}}$.

\item $\expr \equiv c$. 
By definition of $\dyn{c}$ this case falls in the previous. 
\item $\expr \equiv \eif{x}{\expr_1}{\expr_2}$. 
By inversion of the hypothesis
$\hastype{}{\shape{\env}}{x}{\shape{\tref{v}{\tbool}{\_}}}$, 
$\hastype{}{\shape{\env}}{\shape{\expr_1}}{\shape{\typ}}$, and  
$\hastype{}{\shape{\env}}{\shape{\expr_2}}{\shape{\typ}}$.
By inductive hypothesis 
$\hastype{}{\dyn{\env}}{x}{\dyn{\tref{v}{\tbool}{\_}}}$, 
$\hastype{}{\dyn{\env}}{\dyn{\expr_1}}{\dyn{\typ}}$, and  
$\hastype{}{\dyn{\env}}{\dyn{\expr_2}}{\dyn{\typ}}$.
By weakening, Lemma~\ref{lemma:dynwellformed} and the rule \ruletif
$\hastype{}{\dyn{\env}}{\dyn{e}}{\dyn{\typ}}$.
\item $\expr \equiv \elet{x}{\typ}{\expr_x}{\expr'}$
By inversion of the hypothesis 
$\hastype{}{\shape{\env}}{\shape{\expr_x}}{\shape{\typ_x}}$
and
$\hastype{}{\shape{\env, \bind{x}{\typ_x}}}{\shape{\expr'}}{\shape{\typ}}$.
By inductive hypothesis, rule \ruletlet, and Lemma~\ref{lemma:dynwellformed},
$\hastype{}{\dyn{\env}}{\dyn{e}}{\dyn{\typ}}$.
\item $\expr \equiv \elettyp{x}{\typ}{\expr_x}{\expr}$
By inversion of the hypothesis 
$\hastype{}{\shape{\env}}{\shape{\expr_x}}{\shape{\typ'_x}}$
and
$\hastype{}{\shape{\env, \bind{x}{\typ_x}}}{\shape{\expr'}}{\shape{\typ}}$.
By hypothesis, Lemma~\ref{lemma:dynsubtype} and rule \ruletsub
$\hastype{}{\shape{\env}}{\shape{\expr_x}}{\shape{\typ_x}}$.
By inductive hypothesis, rule \ruletspec, and Lemma~\ref{lemma:dynwellformed},
$\hastype{}{\dyn{\env}}{\dyn{e}}{\dyn{\typ}}$.
\end{itemize}
\end{proof}

\begin{theorem}[Embedding of Unrefined Terms]
If all the refinements in constants and user provided specifications are local, 
then if $\hastype{}{\shape{\env}}{\shape{\expr}}{\shape{\typ}}$, 
then $\ginfer{\dyn{\env}}{\dyn{\expr}}{\iquals}\not = \emptyset$.
\end{theorem}
\begin{proof}
Since by Lemma~\ref{lemma:dynamic}
the theorem is proved by completeness of our inference algorithm, 
\ie Theorem~\ref{theorem:completeness}. 
\end{proof}

\paragraph{(iii) Static Gradual Guarantee}

\begin{definition}[Precision of Gradual Types]
$\gless{\gltyp_1}{\gltyp_2}$ \textit{iff}
$\lconcrete{\gltyp_1} \subseteq  \lconcrete{\gltyp_2}$.
\end{definition}

\begin{theorem}[Static Gradual Guarantee]
If 
$\glenv_1\sqsubseteq\glenv_2$ and
$\glexpr_1\sqsubseteq\glexpr_2$,
then 
for every 
$\gltyp_{1i}\in\ginfer{\glenv_1}{\glexpr_1}{\iquals}$
there exists \gless{\gltyp_{1i}}{\gltyp_{2i}} so that 
$\gltyp_{2i}\in \ginfer{\glenv_2}{\gltyp_2}{\iquals}$.
\end{theorem}

\begin{proof}
Since $\gltyp_{1i} \in \ginfer{\glenv_1}{\glexpr_1}{\iquals}$
then $\exists \ltsol$ so that

\noindent
$
\begin{array}{rl}
\quad(1)&\gltyp_{1i} = \applysub{\gltyp_1}{\ltsol}\hfill \\
(2)&\texttt{Just}\ \ltsol \in \qgsolve{(\lambda \kappa. \quals)}{\gliquid{C}_1} \\
(3)&(\texttt{Just}\ \gltyp_1,\glC_1) = \texttt{Cons}\ \glenv_1\ \glexpr_1 \\
&\\
\multicolumn{2}{l}{\text{From}\ (2),}\\
&\\
(4)&\exists \algradualinstance{C_1}.
\texttt{Just}\ \ltsol \in \solve{(\lambda \kappa. \quals)}{\lC_1} \\
&\\
\multicolumn{2}{l}{\text{Since \texttt{Cons} is preserves}\ \egrad}\\
&\\
(5)&(\texttt{Just}\ \gltyp_2,\glC_2) = \texttt{Cons}\ \glenv_2\ \glexpr_2 \\
(6)&\gless{\gltyp_1}{\gltyp_2}\\
(7)& \gless{\glC_1}{\glC_2}\\
&\\
\multicolumn{2}{l}{\text{By}\ (4) \text{ and } (7)\ \text{we get}}\\
&\\
(8)&\lC_1 \in \lconcrete{\glC_2} \\
&\\
\multicolumn{2}{l}{\text{So,}}\\
&\\
(9)&\texttt{Just}\ \ltsol \in \qgsolve{(\lambda \kappa. \quals)}{\gliquid{C}_2} \\
&\\
\multicolumn{2}{l}{\text{By}\ (5) \text{ and } (9)\ \text{we get}}\\
&\\
(10)&\applysub{\gltyp_2}{\ltsol} \in \ginfer{\glenv_2}{\glexpr_2}{\iquals} \\
&\\
\multicolumn{2}{l}{\text{By}\ (6), (10)\ \text{and since } \ltsol\ \text{preserves}\ \egrad}\\
&\\
&\gless{\applysub{\gltyp_1}{\ltsol}}{\applysub{\gltyp_2}{\ltsol}} \\
\end{array}
$

\end{proof}
 
\end{document}